\documentclass{article}
\pdfoutput=1
\usepackage[margin=1in]{geometry}
\usepackage{amsmath, amssymb, amsthm}
\usepackage{multirow}
\usepackage{enumitem}
\usepackage{color}
\usepackage{algpseudocode}
\usepackage{algorithmicx}
\usepackage{algorithm}

\newcommand{\eps}{\epsilon}
\newcommand{\R}{\ensuremath\mathbb{R}}
\newcommand{\N}{\ensuremath\mathbb{N}}
\newcommand{\Z}{\ensuremath\mathbb{Z}}

\DeclareRobustCommand{\stir}{\genfrac\{\}{0pt}{}}

\newmuskip\pFqmuskip

\newcommand*\pFq[6][8]{%
  \begingroup 
  \pFqmuskip=#1mu\relax
\pFqcomma
  {}_{#2}F_{#3}{\left(\genfrac..{0pt}{}{#4}{#5};#6\right)}%
  \endgroup
}
\newcommand{\pFqcomma}{\mskip\pFqmuskip}

\DeclareMathOperator{\rk}{rank}

\DeclareMathOperator{\tr}{tr}
\DeclareMathOperator{\poly}{poly}

\DeclareMathOperator*{\E}{\mathbb{E}}

\DeclareMathOperator*{\Unif}{Unif}
\DeclareMathOperator*{\diag}{diag}
\DeclareMathOperator*{\argmax}{\arg\!\max}

\allowdisplaybreaks

\newtheorem{theorem}{Theorem}
\newtheorem{lemma}{Lemma}

\theoremstyle{definition}
\newtheorem{definition}{Definition}
\theoremstyle{remark}

\title{On Approximating Functions of the Singular Values in a Stream\footnote{A preliminary version is to appear in the \textit{Proceedings of STOC} 2016.}}
\author{
	Yi Li\footnote{Supported by ONR grant N00014-15-1-2388 when the author was at Harvard University.}\\
	Nanyang Technological University\\
	\texttt{yili@ntu.edu.sg}
\and
	David P. Woodruff\footnote{Supported in part 
by the XDATA  program  of  the  Defense  Advanced  Research  Projects  
Agency (DARPA),  administered  through  Air  Force  Research  Laboratory  
contract  FA8750-12-C-0323.}\\
	IBM Almaden Research Center\\
	\texttt{dpwoodru@us.ibm.com}
}
\date{}

\begin{document}
\maketitle

\begin{abstract}
For any real number $p > 0$, we nearly completely 
characterize the space complexity of estimating $\|A\|_p^p = \sum_{i=1}^n \sigma_i^p$ for
$n \times n$ matrices $A$ in which each row and each column has $O(1)$ non-zero entries and 
whose entries are presented one at a time in a data stream model.
Here the $\sigma_i$ are the singular values of $A$, and when $p \geq 1$, $\|A\|_p^p$ is the 
$p$-th power of the Schatten $p$-norm. We show that when $p$ is not an even integer, to obtain
a $(1+\epsilon)$-approximation to $\|A\|_p^p$ with constant probability, any $1$-pass algorithm 
requires $n^{1-g(\epsilon)}$ bits of space, where $g(\epsilon) \rightarrow 0$
as $\epsilon \rightarrow 0$ and $\epsilon > 0$ is a constant independent of $n$. 
However, when $p$ is an even integer, we give an upper bound of $n^{1-2/p} \poly(\epsilon^{-1}\log n)$
bits of space, which holds even in the turnstile data stream model. The latter is optimal up to 
$\poly(\epsilon^{-1} \log n)$ factors. 

Our results considerably strengthen lower bounds in previous work for arbitrary (not necessarily sparse) 
matrices $A$: 
the previous best lower bound was $\Omega(\log n)$ for $p\in (0,1)$, $\Omega(n^{1/p-1/2}/\log n)$ for $p\in [1,2)$ and $\Omega(n^{1-2/p})$ for $p\in (2,\infty)$. We note for $p \in (2, \infty)$, while our lower
bound for even integers is the same, for other $p$ in this range our lower bound is 
$n^{1-g(\epsilon)}$, which is considerably stronger than the previous $n^{1-2/p}$ 
for small enough constant $\epsilon > 0$.  
We obtain similar near-linear lower bounds for Ky-Fan norms, SVD entropy, eigenvalue shrinkers, and M-estimators, many of which could have been solvable in logarithmic space prior to our work. 
\end{abstract}

\section{Introduction}
In the data stream model, there is an underlying vector $x \in \mathbb{Z}^n$ which 
undergoes a sequence of additive updates to its coordinates. 
Each update has the form $(i, \delta) \in [n] \times \{-m, -m+1, \ldots, m\}$ (where $[n]$ denotes $\{1,\dots,n\}$), 
and indicates that $x_i \leftarrow x_i + \delta$. The algorithm maintains a small summary of $x$ while
processing the stream. At the end of the stream it should succeed in approximating a pre-specified 
function of $x$ with constant probability. The
goal is often to minimize the space complexity of the algorithm while processing the stream.
We make the standard simplifying assumption that $n,m$, and the length of the stream are polynomially related. 

A large body of work has focused on characterizing which functions $f$ it is possible to 
approximate $f(x) = \sum_{i = 1}^n f(x_i)$ using a polylogarithmic (in $n$) amount of space. 
The first class of functions studied were the $\ell_p$ norms $f(x_i) = |x_i|^p$, 
dating back to work of Alon, Matias, and Szegedy \cite{ams99}. 
For $p \leq 2$ it is possible to obtain any constant factor approximation using $\tilde{\Theta}(1)$ 
bits of space \cite{i06,knw10}, while for $p > 2$ the bound is 
$\tilde{\Theta}(n^{1-2/p})$ \cite{cks03,bjks04,iw05,ako11,g12,lw13,bksv14,g15}, where $\tilde{f} = f \cdot \poly(\log(f))$.  
Braverman and Ostrovsky later developed a zero-one law for monotonically non-decreasing $f$ for 
which $f(0) = 0$, showing that if $f$ has at most quadratic growth and does not have large ``local jumps'', 
then a constant factor approximation to $f(x)$ can be computed in $\tilde{O}(1)$ space
 \cite{bo10a}. Moreover, if either condition 
is violated, then there is no polylogarithmic space algorithm. This was extended by Braverman and Chestnut to 
periodic and to decreasing $f$ \cite{bc14,bc15}. 
Characterizations were also given in the related sliding window model \cite{bor15}.
 Recently, Bravermen et al.\ gave conditions nearly characterizing all $f$ computable in a constant number 
of passes using $n^{o(1)}$ space \cite{bcwy15}. 

Despite a nearly complete understanding of which functions $f$ one can approximate $\sum_{i=1}^n f(x_i)$ for a vector $x$ using small space in a stream, 
little is known about estimating functions of an $n \times n$ {\it matrix} $A$ presented in a stream. Here,  
an underlying $n \times n$ matrix $A$ undergoes a sequence 
of additive updates to its entries. Each update has the form
$(i, j, \delta) \in [n] \times [n] \times \{-m,-m+1,\dots,m\}$ and indicates that $A_{i, j} \leftarrow
A_{i, j} + \delta$. Every matrix $A$ can be expressed in its singular value decomposition 
as $A = U \Sigma V^T$, where $U$ and $V$ are orthogonal $n \times n$ matrices, and $\Sigma$ is a non-negative 
diagonal matrix with diagonal entries $\sigma_1 \geq \sigma_2 \geq \cdots \geq \sigma_n$, which are the singular values of $A$. 
We are interested in functions which do not depend on the bases $U$ and $V$, but rather only on the spectrum (singular values)
of $A$. 
These functions have the same value under any (orthogonal) change of basis. 

The analogue of the functions studied for vectors are functions of the form $\sum_{i=1}^n f(\sigma_i)$. 
Here, too, we can take $f(\sigma_i) = \sigma_i^p$, in which case $\sum_{i=1}^n f(\sigma_i)$ is the $p$-th power of {\it Schatten} $p$-norm $\|A\|_p^p$ of $A$. 
When $p = 0$, interpreting $0^0$ as $0$ this is the {\it rank} of $A$, which has been studied in the data stream 
\cite{cw09,BS15} and property testing models \cite{ks03,lww14}. 
When $p = 1$, this is the nuclear or trace norm\footnote{The trace norm is not to be confused with
the trace. These two quantities only coincide if $A$ is positive semidefinite.}, with applications to differential privacy \cite{hlm10,lm12} 
and non-convex optimization \cite{cr12,dtv11}. When $p = 2$ this is the Frobenius norm, while for large $p$, this sum approaches the $p$-th power of the operator norm $\sup_{x:\|x\|_2=1} \|Ax\|_2$. 
Such norms are useful in geometry and linear algebra, see, e.g.,~\cite{w14}. 
The Schatten $p$-norms also arise in the context of estimating 
$p$-th moments of a multivariate 
Gaussian matrix in which the components are independent but not of the same variance, see, e.g., \cite{kv16}. 
The Schatten $p$-norms have been studied in the sketching model~\cite{LNW14}, 
and upper bounds there imply upper bounds for streaming. 
Fractional Schatten-$p$ norms of Laplacians were studied by Zhou~\cite{z08} and Bozkurt et al.\ \cite{bozkurt2012sum}. 
We refer the reader to \cite{s09} for applications of the case $p = 1/2$,
which is the Laplacian-energy-like (LEL) invariant of a graph. 

There are a number of other functions $\sum_{i=1}^n f(\sigma_i)$ of importance. One is the ``SVD entropy'' for 
which $f(\sigma_i) = \frac{\sigma_i^2}{\|A\|_F^2} \log_2 \frac{\|A\|_F^2}{\sigma_i^2}$. See the work
of Alter et al. for its foundational applications in genome processing \cite{a00}.   
In insertion-only streams
$\|A\|_F^2$ can be computed exactly, so one can set $f(\sigma_i) = \sigma_i^2 \log_2 \frac{1}{\sigma_i^2}$, from
which given $\|A\|_F^2$ and an approximation to $\sum_{i=1}^n f(\sigma_i)$, one can approximate the SVD entropy. 
Other functions are motivated from regularized
low rank approximation, where one computes the optimal {\it eigenvalue shrinkers} for different
loss functions, such as the Frobenius, operator, and nuclear norm losses \cite{GD14}. For example, 
for Frobenius norm loss, $f(x) = \frac{1}{x}\sqrt{(x^2-\alpha-1)^2-4\alpha}$ for 
$x \geq 1 + \sqrt{\alpha}$, and $f(x) = 0$ otherwise, for a given parameter $\alpha$. 

Other applications include low rank approximation
with respect to functions on the singular values that are not norms, 
such as Huber or Tukey loss functions, which 
could find more robust low dimensional subspaces as solutions; we discuss these functions more in Section~\ref{sec:app}. 

\subsubsection*{Our Contributions} 
The aim of this work is to obtain the first sufficient criteria in the streaming model for 
functions of a matrix spectrum. Prior to our work we did not even know the complexity of most of the
problems we study even in the {\it insertion-only} data stream model 
in which each coordinate
is updated at most once in the stream, and even when $A$ is promised to be \emph{sparse},
i.e., it has only $O(1)$ non-zero entries per row and column. Sparse matrices 
have only a constant factor more entries than diagonal matrices, and the space 
complexity of diagonal matrices is well-understood since it corresponds to that for vectors. 
As a main application, we considerably strengthen the
known results for approximating Schatten $p$-norms. We stress that the difficulty with 
functions of a matrix spectrum is that updates to the matrix entries often affect the singular values
in subtle ways. 

{\it The main qualitative message of this work is that for approximating Schatten $p$-norms up to a sufficiently
small constant factor, for any positive real number 
$p$ which is not an even integer, almost $n$ bits of space is necessary. Moreover, this holds even for matrices
with $O(1)$ non-zero entries per row and column, and consequently is tight for such matrices. It also holds
even in the insertion-only model. Furthermore, for even
integers $p$, we present an algorithm achieving an arbitrarily small constant factor approximation 
for any matrix with $O(1)$ non-zero entries per row and column which achieves 
$\tilde{O}(n^{1-2/p})$ bits of space. Also,
$\Omega(n^{1-2/p})$ bits of space is necessary for even integers $p$, even with $O(1)$ non-zero entries per
row and column and even if all entries are absolute constants independent of $n$. 
Thus, for $p$-norms, there is a substantial
difference in the complexity in the vector and matrix cases:
in the vector case the complexity is logarithmic for $p \leq 2$ and grows as $n^{1-2/p}$ for $p \geq 2$, while in the matrix case the complexity is always almost $n$ bits unless $p$ is an even integer! Furthermore, for each even integer
$p$ the complexity is $\tilde{\Theta}(n^{1-2/p})$, 
just as in the vector case. Note that our results show a ``singularity'' at $p = 2 \pm o(1)$, 
which are the only values of $p$ for which $O(\log n)$ bits of space is possible.}

We now state our improvements over prior work more precisely.  
Henceforth in this section, the approximation parameter $\epsilon$ is a constant (independent of $n$), 
and $g(\epsilon) \rightarrow 0$ as $\epsilon \rightarrow 0$. The number of non-zero entries of $A$ is denoted by nnz$(A)$.
\begin{theorem}(Lower Bound for Schatten $p$-Norms)\label{thm:prior}
Let $p \in [0, \infty) \setminus 2\Z$. 
Any randomized data stream algorithm which outputs, with constant error probability, a $(1+\epsilon)$-approximation to the Schatten $p$-norm of an $n \times n$ matrix $A$ 
requires $\Omega(n^{1-g(\epsilon)})$ bits of space. This holds even if nnz$(A) = O(n)$.  
\end{theorem}
When $p$ is an even integer, Theorem \ref{thm:prior} does not apply and we instead prove the following theorem.
\begin{theorem}(Upper Bound for Schatten $p$-norms for Even Integers $p$)\label{thm:upperArxiv}
Let $p \in 2\Z$ be a positive even integer. For matrices $A$ with $O(1)$ non-zero entries per row and per column, Algorithm~\ref{alg:even_p} (on page~\pageref{alg:even_p}) returns a value that is a $(1+\epsilon)$-approximation to $\|A\|_p^p$ with constant probability, using $O(n^{1-2/p}\poly(1/\epsilon,\log n))$ bits of space.
\end{theorem}
Theorem \ref{thm:upperArxiv} is optimal up to $\poly(1/\epsilon, \log n)$ factors for matices with $O(1)$ non-zeros per row and column, that is,
it matches known lower bounds for vectors (which can be placed on the diagonal of a matrix). We also show an $\Omega(n^{1-2/p})$ lower bound
for even integers $p$ even if all non-zeros of $A$ are promised to be constants independent of $n$. This is a slight strengthening over applying
the lower bound for vectors since the hard instances for vectors have entries which grow with $n$ (roughly as $n^{1/p}$). 

We obtain similar lower bounds as in Theorem \ref{thm:prior} for estimating the Ky-Fan $k$-norm, which is defined to be the sum of the $k$ 
largest singular values, and has applications
to clustering and low rank approximation \cite{Xia,d14}. 
Interestingly, these norms do not have the form $\sum_{i=1}^n f(\sigma_i)$ but rather 
have the form $\sum_{i=1}^k f(\sigma_i)$, yet our framework is robust enough 
to handle them. In the latter case,  we have the following general result for strictly monotone $f$:
\begin{theorem}\label{thm:partial_singular_values}
Let $\alpha\in (0,1/2)$ and $f$ be strictly monotone with $f(0)=0$. There exists a constant $\epsilon>0$ such that for all sufficiently large $n$ and $k\leq \alpha n$, any data stream algorithm which outputs a $(1+\epsilon)$-approximation to $\sum_{i=1}^k f(\sigma_i(A))$ of an $n \times n$ matrix $A$ 
requires $\Omega(n^{1+\Theta(1/\ln\alpha)})$ space. This holds even if nnz$(A) = O(n)$.  
\end{theorem}
We summarize prior work on Schatten $p$-norms and Ky-Fan $k$-norms and its relation to our results 
in Table~\ref{tab:results}. The previous bounds for Ky-Fan norms come from planting a hard instances of the set disjointness communication 
problem on the diagonal
of a diagonal matrix
(where each item is copied $k$ times) \cite{ks92,r92}, or from a Schatten $1$-lower bound on $k \times k$ matrices padded with zeros \cite{AKR15}. 

\begin{table*}[t!]
\vspace{2mm}
\centering{
\begin{tabular}{|c|c|c|c|}
\hline
& & \multicolumn{2}{|c|}{Space complexity in bits} \\
\cline{3-4}
& & Previous lower bounds  & Our lower bounds \\
\hline
\multirow{5}{*}{\parbox[c][][c]{1.25cm}{Schatten \\
$p$-norm}} & $p \in (2,\infty) \cap 2\Z$ & $n^{1-2/p}$ \cite{g09,j09} & \\
\cline{2-4}
& $p \in (2,\infty)\setminus 2\Z$ & $n^{1-2/p}$ \cite{g09,j09} & $n^{1-g(\epsilon)}$\\
\cline{2-4}
& $p \in [1,2)$ & $\frac{n^{1/p-1/2}}{\log n}$ \cite{AKR15}  & $n^{1-g(\epsilon)}$\\
\cline{2-4}
& $p \in (0,1)$ & $\log n$ \cite{knw10} & $n^{1-g(\epsilon)}$\\
\cline{2-4}
& $p = 0$ & $n^{1-g(\epsilon)}$ \cite{BS15} & \\
\hline
\multicolumn{2}{|c|}{Ky-Fan $k$-norm} & $\max\{\frac{n}{k}, \frac{k^{1/2}}{\log k}\}$ \cite{bjks04,AKR15}& $n^{1-g(\epsilon)}$ (any $k$)\\
\hline
\end{tabular}}
\caption{\small{A summary of existing and new lower bounds for $(1+\epsilon)$-approximating Schatten $p$-norms and Ky-Fan $k$-norms, where $\epsilon$ is an arbitrarily small constant. 
The $\Omega$-notation is suppressed. The function $g(\epsilon)\to 0$ as $\epsilon\to 0$ and could depend on the parameters $p$ or $k$ and be different in different rows. We show that the lower bound $n^{1-2/p}$ is tight up to log factors by providing a new upper bound for even integers $p$ and sparse matrices. For even integers we also present a new proof of an $n^{1-2/p}$ lower bound in which
all entries of the matrix are bounded by $O(1)$.}}
\label{tab:results}
\end{table*}
The best previous lower bound for estimating the Schatten $p$-norm up to an arbitrarily 
small constant factor 
for $p \geq 2$ was $\Omega(n^{1-2/p})$, which
is the same for vector $p$-norms. In \cite{LNW14}, an algorithm for {\it even integers} $p \geq 2$
was given, and it works in the data stream model using $O(n^{2-4/p})$ bits of space. See also 
\cite{AN13} for finding large eigenvalues, which can be viewed as an additive approximation
to the case $p = \infty$. For $p \in [1,2)$, the lower bound was 
$\Omega(\frac{n^{1/p-1/2}}{\log n})$ \cite{AKR15}. Their approach is based on non-embeddability, 
and the best lower bound obtainable via this approach is $\Omega(n^{1/p-1/2})$,
since the identity map is an embedding of the Schatten $p$-norm into the Schatten-$2$ norm
with $n^{1/p-1/2}$ distortion, and the latter can be sketched with $O(\log n)$ bits; further
it is unknown if the lower bound of \cite{AKR15} holds for sparse matrices \cite{R15}. 
For $p \in (0,1)$, which is not a norm but still a well-defined quantity, 
the prior bound is only $\Omega(\log n)$, which follows from lower bounds
for $p$-norms of vectors. For $p = 0$, an $\Omega(n^{1-g(\eps)})$ lower bound
was shown for $(1+\eps)$-approximation \cite{BS15}. 
We note that lower bounds for Schatten-$p$ norms in the sketching model,
as given in \cite{LNW14}, do not apply to the streaming model, even given work 
which characterizes ``turnstile'' streaming
algorithms as linear sketches\footnote{In short,
in the sketching model one has a matrix $S$ and one distinguishes $S \cdot X$ from $S \cdot Y$ where $X, Y$ are vectors
(or vectorized matrices) with $X \sim \mu_1$ and $Y \sim \mu_2$ for distributions $\mu_1$ and $\mu_2$. One argues if
$S$ has too few rows, then $S \cdot X$ and $S \cdot Y$ have small statistical distance, but such a statement is not 
true if we first discretize $X$ and $Y$.} \cite{lnw14b}. 

One feature of previous work is that it rules out constant factor approximation for 
a large constant
factor, whereas our work focuses on small constant factor approximation. For vector
norms, the asymptotic complexity in the two cases is the same \cite{knw10,bksv14} 
or the same up to a logarithmic
factor \cite{g12,lw13}. Given the many motivations and extensive work on obtaining $(1+\eta)$-approximation for
vector norms for arbitrarily small $\eta$ \cite{iw03,w04,cm05,gc07,nw10,knw10,knpw11,g12,wz12,lw13}, 
we do not view this as a significant
shortcoming. Nevertheless, this is an interesting open question, which could exhibit another
difference between matrix and vector norms. 

Although Theorem~\ref{thm:prior} makes significant progress on Schatten $p$-norms,
and is nearly optimal for sparse matrices (i.e., matrices with $O(1)$ non-zero entries per row and column), 
for dense matrices our bounds are off by a quadratic factor. That is, for $p$ not an even integer, 
we achieve a lower bound which is almost $n$ bits of space, while the upper bound is a trivial $O(n^2)$ words
of space used to store the matrix.
Therefore, $\tilde{\Theta}(n^{1-2/p})$ is an upper and lower
bound for sparse matrices, while for dense matrices the best upper bound is $O(n^{2-4/p})$ given in \cite{LNW14}.
Thus, in both cases ($p$ an even integer, or not) the upper bound is the square of the current lower bound. Resolving this gap is an
intriguing open question. 

The Schatten $p$-norms capture a wide range of possibilities of growths of more general functions, and we are able
to obtain lower bound for a general class of functions by considering their growth near $0$ (by scaling down our
hard instance) or their growth for large inputs (by scaling up our hard instance). If in either case the function
``behaves'' like a Schatten $p$-norm (up to low order terms), then we can apply our lower bounds for Schatten $p$-norms
to obtain lower bounds for the function. 

\subsubsection*{Technical Overview} 
\noindent\textbf{Lower Bound.} The starting point of our work is \cite{BS15}, which showed 
an $\Omega(n^{1-g(\epsilon)})$ lower bound
for estimating the rank of $A$ up to a $(1+\epsilon)$-factor by using the fact that the rank 
of the Tutte matrix equals twice the size of the maximum matching of the corresponding graph, 
and there are 
lower bounds for estimating the maximum matching size in a stream \cite{vy11}. 

This suggests that lower bounds for approximating matching size could be used more generally for
establishing lower bounds for estimating Schatten $p$-norms. We abandon the use of the Tutte matrix, 
as an analysis of its singular values turns out to be quite involved. Instead, we devise simpler families
of hard matrices which are related to hard graphs for estimating matching sizes. 
Our matrices
are block diagonal in which each block has constant size (depending on $\epsilon$). 
For functions $f(x) = |x|^p$ for $p>0$ not an even integer, 
we show a constant-factor multiplicative gap in the value of $\sum_i f(\sigma_i)$ in the case where
the input matrix is (1) block diagonal in which each block is 
the concatenation of an all-1s matrix and a diagonal matrix with an {\it even number}
of 1s versus (2) block diagonal in which each block is the concatenation of an all-1s matrix and a diagonal matrix
with an {\it odd number} of ones. We call these Case 1 and Case 2. We also refer to the 1s on a diagonal matrix
inside a block as {\it tentacles}. 

The analysis proceeds by looking at a block in which the number of tentacles follows a binomial distribution.
We show that the expected value of $\sum_i f(\sigma_i)$ restricted to a block given that the number of tentacles is even,
differs by a constant factor from the expected value of $\sum_i f(\sigma_i)$ restricted to a block given that the number
of tentacles is odd. Using the hard distributions for matching \cite{bjk04,GKKRW07,vy11}, 
we can group the blocks into independent groups of four matrices and then apply 
a Chernoff bound across the groups 
to conclude that with high probability, $\sum_i f(\sigma_i)$ of the entire matrix in Case 1 differs by a $(1+\epsilon)$-factor from $\sum_i f(\sigma_i)$ of the entire matrix in Case 2. 
This is formalized in Theorem \ref{thm:gap_expectation}. 

The number $k$ of tentacles is subject to a binomial distribution supported on even or odd numbers in 
Case 1 or 2 respectively. Proving a ``gap'' in expectation for a random
even value of $k$ in a block versus a random odd value of $k$ in a block is intractable if the expressions for
the singular values are sufficiently complicated. For example, the singular values of the adjacency matrix 
of the instance in \cite{BS15} for $p = 0$ involve roots of a cubic equation, which poses a great obstacle. 
Instead our hard instance has the advantage that the singular values $r(k)$ are the square roots of the roots 
of a quadratic equation, which are more tractable. The function value $f(r(k))$, viewed as a 
function of the number of tentacles $k$, can be expanded into a power series 
$f(r(k)) = \sum_{s=0}^\infty c_s k^s$, and the difference in expectation in the even and odd cases subject to a
binomial distribution is 
\[
\sum_{k=0}^m (-1)^k\binom{m}{k} f(r(k)) = \sum_{s=0}^\infty c_s \sum_{k=0}^m (-1)^i\binom{m}{k} k^s\\
= \sum_{s=0}^\infty c_s (-1)^m m! \stir{s}{m}\\
= (-1)^m m! \sum_{s=m}^\infty c_s \stir{s}{m}
\]
where $\stir{s}{m}$ is the Stirling number of the second kind and $\stir{s}{m}=0$ for $s<m$, and where
the second equality is a combinatorial identity. 
The problem reduces to analyzing the last series of $s$. 
For $f(x)=|x|^p$ ($p > 0$ not an even integer), 
with our choice of hard instance which we can parameterize by a small constant $\gamma > 0$, 
the problem reduces to showing that $c_s = c_s(\gamma) > 0$ 
for a small $\gamma$, and for all large $s$. However, $c_s(\gamma)$ is complicated and admits the form of 
a hypergeometric polynomial, which can be transformed to a different hypergeometric function $c_s'(\gamma)$ of 
simpler parameters. It has a standard infinite series expansion 
$c_s'(\gamma) = 1 + \sum_{n=0}^\infty a_n \gamma^n$. 
By analyzing the series coefficients, the infinite series can be split into three parts: a head part, a middle part,
and a tail; each can be analyzed separately. 
Roughly speaking, the head term is 
alternating and decreasing and thus dominated by its first term, the tail term has geometrically decreasing 
terms and is also dominated by its first term, which is much smaller than the head, and finally 
the middle term is dominated by the head term. 

The result for $f(x)=x^p$ 
generalizes to functions which are asymptotically $x^p$ near $0$ or infinity, by first scaling the input matrix by a small or a large constant.

\vskip 1ex

\noindent\textbf{A simple $\sqrt{n}$ lower bound.} To illustrate our ideas, here we give a very simple
proof of an $\Omega(\sqrt{n})$ lower bound for any real $p \neq 2$. Consider the three possible blocks
\[ A = \begin{pmatrix}
0 & 1\\
1 & 0
\end{pmatrix},\  
B = \begin{pmatrix}
 1 & 1\\
1 & 0
\end{pmatrix},\ 
 C = 
\begin{pmatrix}
1 & 1\\
1 & 1
\end{pmatrix}.\]
A simple computation shows the two singular values are $(1, 1)$ for $A$, are 
$((\sqrt{5}+1)/2, (\sqrt{5}-1)/2)$ for $B$, and $(2, 0)$ for $C$. Our reduction above 
from the Boolean Hidden Matching Problem implies for Schatten $p$-norm estimation, we
get an $\Omega(\sqrt{n})$ lower bound for $c$-approximation for a small enough constant $c > 1$,
provided 
$$\frac{1}{2} \cdot (1^p + 1^p) + \frac{1}{2} \cdot 2^p 
\neq \left (\frac{\sqrt{5}+1}{2} \right )^p + \left(\frac{\sqrt{5}-1}{2} \right )^p ,$$
which holds for any real $p \neq 2$.
\vskip 1ex
\noindent\textbf{Upper Bound.} We illustrate the ideas of our upper bound with $p=4$, in which case, 
$\|A\|_4^4 = \sum_{i,j}|\langle a_i,a_j\rangle|^2$, where $a_i$ is the $i$-th row of $A$. 
Suppose for the moment that every row $a_i$ had the same norm $\alpha = \Theta(1)$. It would
then be easy to estimate $n \alpha^4 = \sum_i |\langle a_i, a_i \rangle|^2 = \Theta(n)$ 
just by looking at the norm of a single row. Moreover,
by Cauchy-Schwarz, $\alpha^4 = \|a_i\|^4 \geq |\langle a_i, a_j \rangle|^2$ 
for all $j \neq i$. 
Therefore in order for
$\sum_{i \neq j} |\langle a_i, a_j \rangle|^2$ to ``contribute'' to $\|A\|_4^4$, its value must be $\Omega(n \alpha^4)$,
but since each summand is upper-bounded by $\alpha^4$, there must be $\Omega(n)$ non-zero terms. 
It follows that if we sample $O(\sqrt{n})$ rows uniformly
and in their entirety, by looking at all $O(n)$ pairs $|\langle a_i, a_j \rangle|^2$ for sampled rows $a_i$ and $a_j$,
we shall obtain $\Omega(1)$ samples of the ``contributing'' pairs $i \neq j$. Using that each row and column
has $O(1)$ non-zero entries, 
this can be shown to be enough to
obtain a good estimate to $\|A\|_4^4$ and it uses $O(\sqrt{n} \log n)$ bits of space. 

In the general situation where the rows of $A$ have differing norms, we need to sample them proportional to 
their squared $2$-norm. Also, it is not possible to obtain the sampled rows $a_i$ in their entirety, but we can
obtain noisy approximations to them. We achieve this by adapting known algorithms for $\ell_2$-sampling in a data
stream \cite{MW10,ako11,JST11}
, and using our conditions that each row and each column of $A$ have $O(1)$ non-zero entries. 
Given rows $a_{i}$ and $a_{j}$, one can verify that 
$|\langle a_{i},a_{j}\rangle|^2 \frac{\|A\|_F^4}{\|a_{i}\|_2^2 \|a_{j}\|_2^2}$ is an unbiased estimator of 
$\|A\|_4^4$, and in fact, this is nothing other than importance sampling. It turns out
that also in this more general case, 
only $O(\sqrt{n})$ rows need to be sampled, and we can look at all $O(n)$ pairs of inner products between such
rows. 

\section{Preliminaries}\label{sec:prelim}

\paragraph{Notation}
Let $\R^{n\times d}$ be the set of $n\times d$ real matrices. 
We write $X\sim \mathcal{D}$ for a random variable $X$ subject to a probability distribution $\mathcal{D}$. Denote the uniform distribution on a set $S$ (if it exists) by $\Unif(S)$.

We write $f\gtrsim g$ (resp.\ $f\lesssim g$) if there exists a constant $C>0$ such that $f\geq Cg$ (resp.\ $f\leq Cg$). Also we write $f\simeq g$ if there exist constants $C_1 > C_2 > 0$ such that $C_2 g \leq f \leq C_1 g$. For the notations hiding constants, such as $\Omega(\cdot)$, $O(\cdot)$, $\lesssim$, $\gtrsim$, we may add subscripts to highlight the dependence, for example, $\Omega_a(\cdot)$, $O_a(\cdot)$, $\lesssim_a$, $\gtrsim_a$ mean that the hidden constant depends on $a$.

\paragraph{Singular values and matrix norms}
Consider a matrix $A\in \R^{n \times n}$. Then $A^TA$ is a positive semi-definite matrix. The eigenvalues of $\sqrt{A^TA}$ are called the singular values of $A$, denoted by $\sigma_1(A)\geq \sigma_2(A)\geq \cdots \geq\sigma_n(A)$ in decreasing order. Let $r=\rk(A)$. It is clear that $\sigma_{r+1}(A) = \cdots = \sigma_n(A) = 0$. 
Define $
\|A\|_p = (\sum_{i=1}^r (\sigma_i(A))^p)^{1/p}
$ ($p>0$).
For $p\geq 1$, it is a norm over $\R^{n\times d}$, called the $p$-th \textit{Schatten norm}, over $\R^{n\times n}$ for $p\geq 1$. When $p=1$, it is also called the trace norm or nuclear norm. When $p=2$, it is exactly the Frobenius norm $\|A\|_F$. 
Let $\|A\|_{op}$ denote the operator norm of $A$ when treating $A$ as a linear operator from $\ell_2^n$ to $\ell_2^n$. It holds that $\lim_{p\to\infty} \|A\|_p = \sigma_1(A) = \|A\|_{op}$.

The Ky-Fan $k$-norm of $A$, denoted by $\|A\|_{F_k}$, is defined as the sum of the largest $k$ singular values: $\|A\|_{F_k} = \sum_{i=1}^k \sigma_i(A)$. Note that $\|A\|_{F_1} = \|A\|_{op}$ and $\|A\|_{F_k} = \|A\|_1$ for $k\geq r$.

\paragraph{Communication Complexity}
We shall use a problem called Boolean Hidden Hypermatching, denoted by $\textsc{BHH}_{t,n}^0$, defined in \cite{vy11}.
\begin{definition}
In the Boolean Hidden Hypermatching Problem $\textsc{BHH}_{t,n}$, Alice gets a Boolean vector $x \in \{0, 1\}^n$ with $n = 2rt$ for some integer $r$ and Bob gets a perfect
$t$-hyper-matching $M$ on the $n$ coordinates of $x$, i.e., each edge has exactly $t$ coordinates, and a binary string $w\in \{0,1\}^{n/t}$. Let $Mx$ denote the vector of length $n/t$ defined as $(\bigoplus_{1\leq i\leq t} x_{M_{1,i}}$, ..., $\bigoplus_{1\leq i\leq t} x_{M_{n/t,i}})$, where $\{(M_{j,1},\dots,M_{j,t})\}_{j=1}^{n/t}$ are edges of $M$. It is promised that either  $Mx \oplus w = \mathbf{1}^{n/t}$ or $Mx \oplus w = \mathbf{0}^{n/t}$. The problem is to return $1$ in the first case and $0$ otherwise.
\end{definition}
They proved that this problem has an $\Omega(n^{1-1/t})$ randomized one-way communication lower bound by proving a lower bound for deterministic protocols with respect to the hard distribution in which $x$ and $M$ are independent and respectively uniformly distributed, and $w=Mx$ with probability $1/2$ and $w=\overline{Mx}$ (bitwise negation of $Mx$) with probability $1/2$. In \cite{BS15}, Bury and Schwiegelshohn defined a version without $w$ and with the constraint that $w_H(x) = n/2$, for which they also showed an $\Omega(n^{1-1/t})$ lower bound. We shall use this version, with a slight modification.
\begin{definition}
In the \emph{Boolean Hidden Hyper-matching Problem} $\textsc{BHH}^0_{t,n}$, Alice gets a Boolean  vector
$x\in \{0,1\}^n$ with $n=4rt$ for some $r\in \N$ and even integer $t$ and $w_{H}(x) = n/2$, Bob gets a perfect $t$-hypermatching $M$ on the $n$ coordinates of $x$, i.e., each edge has exactly $t$ coordinates. We denote by $Mx$ the Boolean vector of length $n/t$ given by $\big(\bigoplus_{i=1}^t x_{M_{1,i}},\dots,\bigoplus_{i=1}^t x_{M_{n/t,i}}\big)$, where $\{(M_{j,1},\dots,M_{j,t})\}_{j=1}^{n/t}$ are the edges of $M$. It is promised that either $Mx = \mathbf{1}^{n/t}$ or $Mx = \mathbf{0}^{n/t}$. The problem is to return $1$ in the first case and $0$ otherwise.
\end{definition}
A slightly modified (yet remaining almost identical) proof as in \cite{BS15} shows that this problem also has an $\Omega(n^{1-1/t})$ randomized one-way communication lower bound. We include the proof here.
\begin{proof}
We reduce $\textsc{BHH}_{t,n}$ to $\textsc{BHH}^0_{t,2n}$. Let $x\in \{0,1\}^n$ with $n=2rt$ for some $r$, $M$ be a perfect $t$-hypermatching on the $n$
coordinates of $x$ and $x\in \{0,1\}^{w/t}$. Define $x' = \begin{pmatrix}x^T & \bar x^T\end{pmatrix}^T$ to be the concatenation of $x$ and $\bar x$ (bitwise negation of $x$). 

Let $\{x_1,\dots, x_t\}\in M$ be the $l$-th hyperedge of $M$. We include two hyperedges in $M'$ as follows. When $w_l=0$, include $\{x_1,\dots,x_t\}$ and $\{\overline{x_1},\overline{x_2},\dots,\overline{x_t}\}$ in $M$; when $w_l=1$, include $\{\overline{x_1},x_2,\dots,x_t\}$ and $\{x_1,\overline{x_2}\dots,\overline{x_t}\}$ in $M'$. The observation is that we flip an even number of bits in the case $w_l=0$ and an odd number of bits when $w_l=1$, and since $t$ is even, this does not change the parity of each set. Therefore $M'x' = 0^{2n}$ if $Mx+w=0^{n/2}$ and $M'x'=1^{2n}$ if $Mx+w=1^{n/2}$. The lower bound then follows from the lower bound for $\textsc{BHH}_{t,n}$.
\end{proof}
When $t$ is clear from context, we shorthand $\textsc{BHH}^0_{t,n}$ as $\textsc{BHH}^0_{n}$.

\paragraph{Special Functions} 
The gamma function $\Gamma(x)$ is defined as
\[
\Gamma(x) = \int_0^\infty x^{t-1}e^{-x}\,dx.
\]
For positive integer $n$ it holds that $\Gamma(n) = (n-1)!$. The definition above can be extended by analytic continuation to all complex numbers except non-positive integers.

The hypergeometric function ${}_pF_q(a_1,a_2,\dots,a_p;b_1,b_2,\dots,b_q;z)$ of $p$ upper parameters and $q$ lower parameters is defined as
\[
\pFq{p}{q}{a_1,a_2,\dots,a_p}{b_1,b_2,\dots,b_q}{z} = \sum_{n=0}^\infty \frac{(a_1)_n(a_2)_n\cdots(a_p)_n}{(b_1)_n(b_2)_n\cdots (b_q)_n} \cdot \frac{z^n}{n!},
\]
where
\[
(a)_n = \begin{cases}
			1, & n = 0;\\
			a(a+1)\cdots(a+n-1), & n > 0.
		\end{cases}	
\]
The parameters must be such that the denominator factors are never $0$. When $a_i = -n$ for some $1\leq i\leq p$ and non-negative integer $n$, the series above becomes terminating and the hypergeometric function is in fact a polynomial in $x$. 
\section{Schatten norms}\label{sec:schatten}
Let $D_{m,k}$ ($0\leq k\leq m$) be an $m\times m$ diagonal matrix with the first $k$ diagonal elements equal to $1$ and the remaining diagonal entries $0$, and let $\mathbf{1}_m$ be an $m$ dimensional vector full of $1$s. Define
\begin{equation}\label{eqn:asymmetric_M}
M_{m,k} = \begin{pmatrix}
				\mathbf{1}_m \mathbf{1}_m^T & 0\\
				\sqrt{\gamma} D_{m,k} & 0
		   \end{pmatrix},
\end{equation}
where $\gamma > 0$ is a constant (which may depend on $m$).

Our starting point is the following theorem. Let $m \geq 2$ and $p_m(k) = \binom{m}{k}/2^{m-1}$ for $0\leq k\leq m$. Let $\mathcal{E}(m)$ be the probability distribution on even integers $\{0,2,\dots,m\}$ with probability density function $p_m(k)$, and $\mathcal{O}(m)$ be the distribution on odd integers $\{1,3,\dots,m-1\}$ with density function $p_m(k)$. We say a function $f$ on square matrices is diagonally block-additive if $f(X)=f(X_1)+\cdots+f(X_s)$ for any block diagonal matrix $X$ with square diagonal blocks $X_1,\dots, X_s$. It is clear that $f(X) = \sum_i f(\sigma_i(X))$ is diagonally block-additive.

\begin{theorem}\label{thm:gap_expectation} Let $t$ be an even integer and $X\in \R^{N\times N}$, where $N$ is sufficiently large. Let $f$ be a function of square matrices that is diagonally block-additive. If there exists $m = m(t)$ such that
\begin{gather}
\E_{q\sim \mathcal{E}(t)} f(M_{m,q}) - \E_{q\sim \mathcal{O}(t)} f(M_{m,q})\neq 0, \label{eqn:gap_expectation}
\end{gather}
then there exists a constant $c=c(t)>0$ such that any streaming algorithm that approximates $f(X)$ within a factor $1\pm c$ with constant error probability must use $\Omega_t(N^{1-1/t})$ bits of space.
\end{theorem}

\begin{proof}
We reduce the problem from the $\textsc{BHH}_{t,n}^0$ problem. Let $n = Nt/(2m)$. For the input of the problem $\text{BHH}_{t,n}^0$, construct a graph $G$ as follows. The graph contains $n$ vertices $v_1,\dots,v_n$, together with $n/t$ cliques of size $m$, together with edges connecting $v_i$'s with the cliques according to Alice's input $x$. These latter edges are called `tentacles'. In the $j$-th clique of size $m$, we fix $t$ vertices, denoted by $w_{j,1},\dots,w_{j,t}$. Whenever $x_i = 1$ for $i = (j-1)(n/t) + r$, we join $v_i$ and $w_{j,r}$ in the graph $G$. 

Let $\mathcal{M}$ be constructed from $G$ as follows: both the rows and columns are indexed by nodes of $G$. For every pair $w,v$ of clique nodes in $G$, let $\mathcal{M}_{w,v} = 1$, where we allow $w=v$. For every `tentacle' $(u,w)$, where $w$ is a clique node, let $\mathcal{M}(u,w) = \sqrt{\gamma}$. Then $\mathcal{M}$ is an $N\times N$ block diagonal matrix of the following form after permuting the rows and columns
\begin{equation}\label{eqn:tutte_matrix_big}
\mathcal{M}_{n,m,t} = \begin{pmatrix}
						 M_{m,q_1} &           &        & \\
						 		    & M_{m,q_2} &        & \\
						 		    &           & \ddots & \\
						 		    &	         &        & M_{m,q_{n/t}}
					  \end{pmatrix},
\end{equation}
where $q_1,\dots,q_{n/t}$ satisfy the constraint that $q_1 + q_2 + \cdots + q_{n/t} = n/2$ and $0\leq q_i \leq t$ for all $i$. It holds that $f(\mathcal{M}_{n,m,t}) = \sum_i f(M_{m,q_i})$.

Alice and Bob will run the following protocol. Alice keeps adding the matrix entries corresponding to `tentacles' while running the algorithm for estimating $f(\mathcal{M})$. Then she sends the state of the algorithm to Bob, who will continue running the algorithm while adding the entries corresponding to the cliques defined by the matching he owns. At the end, Bob outputs which case the input of $\textsc{BHH}_n^0$ belongs to based upon the final state of the algorithm.

From the reduction for $\textsc{BHH}_{t,n}^0$ and the hard distribution of $\textsc{BHH}_{t,n}$, the hard distribution of $\textsc{BHH}_{t,n}^0$ exhibits the following pattern: $q_1,\dots,q_{n/t}$ can be divided into $n/(2t)$ groups. Each group contains two $q_i$'s and has the form $(q,t-q)$, where $q$ is subject to distribution $\mathcal{E}(t)$ or $\mathcal{O}(t)$ depending on the promise. Furthermore, the $q$'s across the $n/(2t)$ groups are independent. The two cases to distinguish are that all $q_i$'s are even (referred to as the \textit{even case}) and that all $q_i$'s are odd (referred to as the \textit{odd case}).

For notational simplicity, let $F_q = f(M_{m,q})$. Suppose that the gap in \eqref{eqn:gap_expectation} is positive. Let $A = \E_{q\sim \mathcal{E}(t)}2(F_q+F_{t-q})$ and $B = 
\E_{q\sim \mathcal{O}(t)} 2(F_q+F_{t-q})$, then $A-B > 0$.
Summing up $(n/2t)$ independent groups and applying a Chernoff bound, with high probability, $f(\mathcal{M})\geq (1-\delta)\frac{n}{2t}A$ in the even case and $f(\mathcal{M})\leq (1+\delta)\frac{n}{2t}A$, where $\delta$ is a small constant to be determined. If we can approximate $f(\mathcal{M})$ up to a $(1\pm c)$-factor, say $X$, then with constant probability, in the even case we have an estimate $X\geq (1-c)(1-\delta)\frac{n}{2t}A$ and in the odd case $X\leq (1+c)(1+\delta)\frac{n}{2t}A$. Choose $\delta=c$ and choose $c < \frac{B-A}{3(B+A)}$. Then there will be a gap between the estimates in the two cases. 
The conclusion follows from the lower bound for the $\textsc{BHH}_n^0$ problem.

A similar argument works when \eqref{eqn:gap_expectation} is negative.
\end{proof}

Our main theorem in this section is the following, a restatement of Theorem~\ref{thm:prior} advertised in the introduction.
\begin{theorem}\label{thm:schatten}\label{THM:SCHATTEN}
Let $p \in (0,\infty)\setminus2\Z$. For every even integer $t$, there exists a constant $c = c(t) > 0$ such that any algorithm that approximates $\|X\|_p^p$ within a factor $1\pm c$ with constant probability in the streaming model must use $\Omega_t(N^{1-1/t})$ bits of space.
\end{theorem}

The theorem follows from applying Theorem~\ref{thm:gap_expectation} to $f(x) = x^p$ and $m = t$ and verifying that \eqref{eqn:gap_expectation} is satisfied. The proof is technical and thus postponed to Section~\ref{sec:schatten_proof}.

For even integers $p$, we change our hard instance to
\[
M_{m,k} = \mathbf{1}_m\mathbf{1}_m^T - I_m + D_{m,k},
\]
where $I_m$ is the $m\times m$ identity matrix.
We then have the following lemma, whose proof is postponed to the end of Section~\ref{sec:even_p_proofs}.
\begin{lemma}\label{lem:even_p}
For $f(x) = x^p$ and integer $p\geq 2$, the gap condition \eqref{eqn:gap_expectation} is satisfied if and only if $t\leq p/2$, under the choice that $m=t$.
\end{lemma}

This yields an $\Omega(n^{1-2/p})$ lower bound, which agrees with the lower bound obtained by injecting the $F_p$ moment problem into the diagonal elements of the input matrix \cite{g09,j09}, but here we have the advantage 
that the entries are bounded by a constant independent of $n$. In fact, for even integers $p$, we show our lower bound is tight up to $\poly(\log n)$ factors for matrices in which every row and column has $O(1)$ non-zero elements by providing an algorithm in Section~\ref{sec:upperbound} for the problem. Hence our matrix construction $M_{m,k}$ will not give a substantially better lower bound. Our lower bound for even integers $p$ also helps us in the setting of general functions $f$ in Section~\ref{sec:app}.

\section{Proof of Theorem~\ref{thm:schatten}}\label{sec:schatten_proof}
\begin{proof}
First we find the singular values of $M_{m,k}$. Assume that $1\leq k\leq m-1$ for now.
\[
M_{m,k}^T M_{m,k} = 
\begin{pmatrix}
	m \textbf{1}_m \textbf{1}_m^T + \gamma D_{m,k} & 0\\
	0 & 0
\end{pmatrix}.
\]
Let $e_i$ denote the $i$-th vector of the canonical basis of $\R^{2m}$. It is clear that $e_1-e_i$ ($i=2,\dots,k$) are the eigenvectors with corresponding eigenvalue $\gamma$, which means that $M_{m,k}$ has $k-1$ singular values of $\sqrt{\gamma}$. Since $M_{m,k}$ has rank $k+1$, there are two more non-zero singular values, which are the square roots of another two eigenvalues, say $r_1(k)$ and $r_2(k)$, of $M_{m,k}^T M_{m,k}$. It follows from $\tr(M_{m,k}^T M_{m,k}) = m + \gamma k$ that $r_1(k) + r_2(k) = m^2 + \gamma$ and from $\|M_{m,k}M_{m,k}\|_F^2 = (m + \gamma)^2 k + (m^2-k) m^2$ that $r_1^2(k) + r_2^2(k) = m^4 + 2\gamma km+\gamma^2$. Hence $r_1(k)r_2(k) = m^2\gamma-km\gamma$. In summary, the non-zero singular values of $M_{m,k}$ are: $\sqrt{\gamma}$ of multiplicity $k-1$, $\sqrt{r_1(k)}$ and $\sqrt{r_2(k)}$, where $r_{1,2}(k)$ are the roots of the following quadratic equation:
\[
x^2 - (m^2 + \gamma) x + (m^2 - km) \gamma = 0.
\]
The conclusion above remains formally valid for $k=0$ and $k=m$. In the case of $k=0$, the matrix $M_{m,0}$ has a single non-zero singular value $m$, while $r_1(k) = m^2$ and $r_2(k) = \gamma$. In the case of $k=m$, the matrix $M_{m,m}$ has singular values $\sqrt{m^2+\gamma}$ of multiplicity $1$ and $\sqrt{\gamma}$ of multiplicity $m-1$, while $r_1(k) = m^2+\gamma$ and $r_2(k) = 0$. Hence the left-hand side of \eqref{eqn:gap_expectation} becomes
\begin{multline*}
\frac{1}{2^{m-1}}\sum_{\text{even }k} \binom{m}{k} \left( (k-1)\gamma^{p/2} + r_1^{p/2}(k) + r_2^{p/2}(k)\right)  -\frac{1}{2^{m-1}}\sum_{\text{odd }k} \binom{m}{k} \left( (k-1)\gamma^{p/2} + r_1^{p/2}(k) + r_2^{p/2}(k)\right)\\
= \frac{1}{2^{m-1}}(G_1 + G_2),
\end{multline*}
where $\gamma^p$ on both sides cancel and
\begin{equation}\label{eqn:tmp_G}
G_i = \sum_{k} (-1)^k \binom{m}{k} r_i^{\frac{p}{2}}(k), \quad i =1,2.
\end{equation}
Our goal is to show that $G_1+G_2\neq 0$ when $p$ is not an even integer. To simplify and to abuse notation, hereinafter in this section, we replace $p/2$ with $p$ in \eqref{eqn:tmp_G} and hence $G_1$ and $G_2$ are redefined to be
\begin{equation}\label{eqn:G}
G_i = \sum_{k} (-1)^k \binom{m}{k} r_i^{p}(k),\quad i=1,2,
\end{equation}
and our goal becomes to show that $G_1 + G_2\neq 0$ for \textit{non-integers} $p$.

Next we choose
\begin{gather*}
r_1(k) = \frac{1}{2}\left(m^2+\gamma + \sqrt{m^4-2\gamma m^2+\gamma^2 + 4\gamma k m}\right),\\
r_2(k) = \frac{1}{2}\left(m^2+\gamma - \sqrt{m^4-2\gamma m^2+\gamma^2 + 4\gamma k m}\right).
\end{gather*}
We claim that they admit the following power series expansion in $k$ (proof deferred to Section~\ref{sec:power_series}),
\[
r_1^p(k) = \sum_{s\geq 0} A_s k^s,\quad r_2^p(k) = \sum_{s\geq 0} B_s k^s,
\]
where for $s\geq 2$, 
\begin{equation}\label{eqn:A_s}
A_s = \frac{ (-1)^{s-1}\gamma^s m^{2p-s}}{s! (m^2-\gamma)^{2s-1}} \sum_{i=0}^{s-1}  (-1)^{i} \binom{s-1}{i} F_{p,s,i} \gamma^{s-i-1} m^{2i}
\end{equation}
and
\begin{equation}\label{eqn:B_s}
B_s =  \frac{(-1)^{s}\gamma^p m^s}{s! (m^2-\gamma)^{2s-1}} \sum_{i=0}^{s-1} (-1)^{i} \binom{s-1}{i} F_{p,s,i} \gamma^i m^{2(s-i-1)},
\end{equation}
and
\[
F_{p,s,i} = \prod_{j=0}^{s-i-1} (p-j) \cdot \prod_{j=1}^{i} (p-2s+j).
\]
We analyse $A_s$ first. It is easy to see that $|F_{s,i}|\leq (2s)^s$ for $s > 2p$, and hence
\begin{align*}
|A_s| &\leq \frac{\gamma^s m^{2p-s}}{s! (m^2-\gamma)^{2s-1}}\sum_{i=0}^{s-1} \binom{s-1}{i} |F_{s,i}| \gamma^{s-i-1} m^{2i}\\
&\leq \frac{\gamma^s m^{2p-s}}{\sqrt{2\pi}s(\frac{s}{e})^s (m^2-\gamma)^{2s-1}} (2s)^s m^{2(s-1)} \left(1 + \frac{\gamma}{m^2}\right)^{s-1}\\
&\leq m^{2p} \frac{2e}{\sqrt{2\pi}s} \frac{\gamma}{m^2(m^2-\gamma)} \left(\frac{4em\gamma}{(m^2-\gamma)^2}\right)^{s-1},
\end{align*}
whence it follows immediately that $\sum_s A_s k^s$ is absolutely convergent. We can apply term after term the identity
\begin{equation}\label{eqn:stirling}
\sum_{k=0}^m \binom{m}{k} k^s (-1)^k = \stir{s}{m} (-1)^m m!,
\end{equation}
where $\stir{s}{m}$ is the Stirling number of the second kind, and obtain that (since $m$ is even)
\[
G_1 = \sum_{s\geq m} \stir{s}{m} m! A_s,
\]
which, using the fact that $\stir{s}{m} m! \leq m^s$, can be bounded as
\[
|G_1| \leq \sum_{s\geq m} m^s |A_s| \leq c_1 m^{2p} \left(\frac{c_2}{m^2}\right)^{m-1}
\]
for some absolute constants $c_1, c_2 > 0$.

Bounding $G_2$ is more difficult, because $B_s$ contains an alternating sum. However, we are able to prove the following critical lemma, whose proof is postponed to Section~\ref{sec:proof_of_B_s>0}.
\begin{lemma}\label{lem:B_s_sign} For any fixed non-integer $p > 0$, one can choose $\gamma_0$ and $m$ such that $B_s$ have the same sign for all $s \geq m$ and all $0 < \gamma < \gamma_0$.
\end{lemma}
Since $\sum_{s\geq m} B_sm^s$ is a convergent series with positive terms, we can apply \eqref{eqn:stirling} to $\sum_s B_s k^s$ term after term, giving the gap contribution from $r_2(k)$ as
\[
G_2 = \sum_{s\geq m} \stir{s}{m} m! B_s.
\]
Let $a_{m,i}$ be the summand in $B_m$, that is,
\[
a_{m,i} = \binom{s-1}{i} F_{s,i} \gamma^i m^{2(s-i-1)}.
\]
Since $p$ is not an integer, $a_{m,i}\neq 0$ for all $i$. Then
\[
r_{m,i} := \frac{a_{m,i}}{a_{m,i-1}} = \frac{m-i-1}{i+1}\cdot \frac{p-2m+i}{p-m+i} \cdot \frac{\gamma}{m^2}.
\]
If we choose $m$ such that $m^2/\gamma \gtrsim ([p]-1)/(p-[p])$ when $p>1$ or $m^2/\gamma \gtrsim 1/(p-[p])$ when $p<1$, it holds that $|r_{m,i}|\leq 1/3$ for all $i$ and thus the sum is dominated by $a_{m,0}$. It follows that
\[
G_2 \geq B_m \gtrsim \frac{\gamma^p m^m}{s! (m^2-\gamma)^{2m-1}}  |a_{m,0}| \gtrsim (p-[p])^2 [p]! \frac{\gamma^p}{(m-\lceil p\rceil-1)^{p-[p]}m^{[p]}}
\]
It follows from Lemma~\ref{lem:B_s_sign} that the above is also a lower bound for $G_2$. Therefore $G_1$ is negligible compared with $G_2$ and $G_1 + G_2\neq 0$. This ends the proof of Theorem~\ref{thm:schatten}.
\end{proof}

\subsection{Proof of Lemma~\ref{lem:B_s_sign}}\label{sec:proof_of_B_s>0}
The difficulty is due to the fact that the sum in $B_s$ is an alternating sum. However, we notice that the sum in $B_s$ is a hypergeometric polynomial with respect to $\gamma/m^2$. This is our starting point.
\begin{proof}[Proof of Lemma~\ref{lem:B_s_sign}]
Let $x=\gamma/m^2$ and write $B_s$ as
\begin{equation}\label{eqn:sum_for_B_s}
B_s = (-1)^{s-1}\frac{\gamma^p m^{3s-2}}{s! (m^2-\gamma)^{2s-1}} \sum_{i=0}^{s-1} (-1)^{i+1} \binom{s-1}{i} F_{s,i} x^i,
\end{equation}
Then
\[
B_s m^s = (-1)^{s-1}\gamma^p \frac{m^{4s-2}}{s! (m^2-\gamma)^{2s-1}} \sum_{i=0}^{s-1} (-1)^{i+1} \binom{s-1}{i} F_{s,i} x^i.
\]
Observe that the sum can be written using a hypergeometric function and the series above becomes
\[
B_s m^s = (-1)^s\gamma^p \frac{1}{s! (1-x)^{2s-1}}\cdot \frac{\Gamma(1+p)}{\Gamma(1+p-s)} 
\cdot \pFq{2}{1}{1-s,1+p-2s}{1+p-s}{x}.
\]
Invoking Euler's Transformation (see, e.g., \cite[p78]{AAR99})
\[
\pFq{2}{1}{a,b}{c}{x} = (1-x)^{c-a-b} \pFq{2}{1}{c-a,c-b}{c}{x}
\]
gives
\begin{equation}\label{eqn:connection}
\pFq{2}{1}{1-s,1+p-2s}{1+p-s}{x} = (1-x)^{2s-1} \pFq{2}{1}{p,s}{p-s+1}{x}.
\end{equation}
Therefore
\begin{equation}\label{eqn:final_series_term}
 B_s m^s = (-1)^s \frac{\gamma^p \Gamma(1+p)}{s!\ \Gamma(1+p-s)} \pFq{2}{1}{p,s}{p-s+1}{x}.
\end{equation}
Since $\Gamma(1+p-s)$ has alternating signs with respect to $s$, it suffices to show that $_2F_1(p,s;p-s+1;x) > 0$ for all $x \in [0, x^\ast]$ and all $s \geq s^\ast$, where both $x^\ast$ and $s^\ast$ depend only on $p$.

Now, we write ${}_2F_1(p,s;p-s+1;x) = \sum_n b_n$, where
\[
b_n = \frac{p(p+1)\cdots(p+n-1)\cdot s(s+1)\cdots(s+n-1)}{(1+p-s)(2+p-s)\cdots (n+p-s)n!} x^n.
\]
It is clear that $b_n$ has the same sign for all $n\geq s-\lceil p\rceil$, and has alternating signs for $n\leq s-\lceil p\rceil$. Consider
\[
\left| \frac{b_n}{b_{n-1}} \right| = \frac{(p+n-1)(s+n-1)}{(p-s+n)n}x.
\]
One can verify that when $n\geq 2s$ and $x\leq 1/10$, $|b_n/b_{n-1}| < 3x \leq 1/3$ and thus $|\sum_{n\geq 2s} b_n| \leq \frac{3}{2}|b_{2s}|$. 
Also, when $s \geq 3p$ is large enough, $x\leq 1/10$ and $n\leq s/2$. It holds that $|b_n/b_{n-1}| < 1$ and thus $\{|b_n|\}$ is decreasing when $n\leq s/2$. (In fact, $\{|b_n|\}$ is decreasing up to $n = \frac{1-x}{1+x}s + O(1)$.) Recall that $\{b_n\}$ has alternating signs for $n\leq s/2$, and it follows that
\[
0\leq \sum_{2\leq n\leq s/2} b_n \leq b_2.
\]
Next we bound $\sum_{s/2 < n < 2s} b_n$. Let $n^\ast = \argmax_{s/2<n<2s} |b_n|$. When $n^\ast \leq s-\lceil p\rceil $,
\begin{align*}
\left|\sum_{s/2 < n < 2s} b_n\right|&\leq \frac{3}{2}s |b_{n^\ast}|\\
&\leq \frac{3}{2}s \frac{p(p+1)\cdots(p+n^\ast)}{n^\ast!} \frac{(s-[p]-n^\ast)!}{(s-[p]-1)!} \frac{(s+n^\ast-1)!}{s!} x^{n^\ast}\\
&\leq \frac{3}{2} s (n^\ast)^{p} \frac{\binom{s+n^\ast-1}{s}}{\binom{s-[p]-1}{s-[p]-n^\ast}} x^{n^\ast}\\
&\leq \frac{3}{2}s \cdot es^p \cdot 4^s \cdot x^{s/2}\\
&\leq x^3,
\end{align*}
provided that $x$ is small enough (independent of $s$) and $s$ is big enough. When $n^\ast > s-\lceil p\rceil$, 
\begin{align*}
\left|\sum_{s/2 < n < 2s} b_n\right| &\leq \frac{3}{2}s |b_{n^\ast}|\\
&\leq \frac{3}{2}s \frac{p(p+1)\cdots(p+n^\ast)}{n^\ast!} \frac{(s+n^\ast-1)!}{(s-[p]-1)!(n^\ast-s+[p]-1)!s!} x^{n^\ast}\\
&\leq \frac{3}{2} s^2 (n^\ast)^{p} \binom{s+n^\ast-1}{s-[p],n^\ast-s+[p]-1,s} x^{n^\ast}\\
&\leq \frac{3}{2} s^2 \cdot e(2s)^p \cdot 3^{3s-1} \cdot x^{s-\lceil p\rceil}\\
&\leq x^3.
\end{align*}
provided that $x$ is small enough (independent of $s$) and $s$ is big enough. Similarly we can bound, under the same assumption on $x$ and $s$ as above, that $|b_{2s}| \leq x^3$. Therefore $|\sum_{n > s/2} b_n| \leq Kx^3$ for some $K$ and sufficiently large $s$ and small $x$, all of which depend only on $p$.

It follows that
\begin{align*}
\pFq{2}{1}{p,s}{p-s+1}{x}&\geq1 - \frac{ps}{s-p-1}x  - \sum_{2\leq n\leq s/2} b_n - \left|\sum_{n > s/2} b_n\right|\\
&\geq 1 - \frac{ps}{s-p-1}x - b_2 - Kx^3\\
&\geq 1 - \frac{ps}{s-p-1}x - \frac{p(p+1)s(s+1)}{2(s-p-1)(s-p-2)}x^2 - Kx^3\\
&>0
\end{align*}
for sufficiently large $s$ and small $x$ (independent of $s$).

The proof of Lemma~\ref{lem:B_s_sign} is now complete.
\end{proof}

\subsection{Proof of Power Series Expansion}\label{sec:power_series}
\begin{proof}[Proof of Claim]
We first verify the series expansion of $r_1(k)$. It is a standard result that for $|x|\leq 1/4$,
\[
\frac{1+\sqrt{1-4x}}{2} = 1 - \sum_{n=1}^\infty C_{n-1} x^n, \ 
\frac{1-\sqrt{1-4x}}{2} = \sum_{n=1}^\infty C_{n-1} x^n,
\]
where $C_n = \frac{1}{n+1}\binom{2n}{n}$ is the $n$-th Catalan number. Let $x=-\gamma k m/(m^2-\gamma)^2$, we have
\begin{align*}
r_1(k) &= m^2\frac{1+\sqrt{1-4x}}{2} + \gamma \frac{1-\sqrt{1-4x}}{2} \\
&= m^2 - (m^2-\gamma) \sum_{n=1}^\infty C_{n-1} x^n \\
&= m^2 - \sum_{n=1}^\infty C_{n-1} \frac{(-1)^n \gamma^n k^n m^n}{(m^2-\gamma)^{2n-1}}
\end{align*}
Applying the generalized binomial theorem,
\begin{align*}
&\quad r_1(k)^p \\
&= m^{2p} + \sum_{i=1}^\infty \binom{p}{i} (-1)^i m^{2(p-i)} \left(\sum_{n=1}^\infty C_{n-1} \frac{(-1)^n \gamma^n k^n m^n}{(m^2-\gamma)^{2n-1}}\right)^i\\
&=m^{2p}  + \sum_{i=1}^\infty  \binom{p}{i} (-1)^i  m^{2(p-i)}  \sum_{n_1,\dots,n_i\geq 1}  \frac{\prod_{j=1}^i C_{n_j-1} (-k\gamma m)^{\sum_j n_j}}{(m^2-\gamma)^{2\sum_j n_j - i}}\\
&=m^{2p} + \sum_{s=1}^\infty  \sum_{i=1}^s \binom{p}{i} m^{2(p-i)}  \frac{(-k\gamma m)^s}{(m^2-\gamma)^{2s - i}}  \sum_{\substack{n_1,\dots,n_i\geq 1\\ n_1+\cdots+n_i = s}}  \prod_{j=1}^i C_{n_j-1},
\end{align*}
where we replace $\sum_j n_j$ with $s$. It is a known result using the Lagrange inversion formula that (see, e.g., \cite[p128]{SF95})
\[
\sum_{\substack{n_1,\dots,n_i\geq 1\\ n_1+\cdots+n_i = s}}  \prod_{j=1}^i C_{n_j-1} = \frac{i}{s}\binom{2s-i-1}{s-1}
\]
Hence (replacing $i$ with $i+1$ in the expression above)
\begin{equation}\label{eqn:D_{p,s}}
A_{s} = \frac{(-1)^{s+1} \gamma^s m^{2p}}{(m^2-\gamma)^{2s-1}}\sum_{i=0}^{s-1} (-1)^{i} \binom{p}{i+1} \frac{i+1}{s}\binom{2s-i-2}{s-1} m^{s-2(i+1)}(m^2-\gamma)^i
\end{equation}
To see that \eqref{eqn:D_{p,s}} agrees with \eqref{eqn:A_s}, it suffices to show that
\[
\sum_{i=0}^{s-1}  (-1)^{i} \binom{s-1}{i} F_{p,s,i} \gamma^{s-i-1} m^{2i-s} = 
s! \sum_{i=0}^{s-1} (-1)^{i} \binom{p}{i+1} \frac{i+1}{s}\binom{2s-i-2}{s-1} m^{s-2(i+1)}(m^2-\gamma)^i
\]
Comparing the coefficients of $\gamma^j$, we need to show that
\[
(-1)^{s-1}\binom{s-1}{j} F_{p,s,j,s-j-1} 
= s! \sum_{i=j}^{s-1} (-1)^i\binom{p}{i+1}\frac{i+1}{s}\binom{2s-i-2}{s-1}\binom{i}{j}
\]
Note that both sides are a degree-$s$ polynomial in $p$ with head coefficient $(-1)^{s-1}$, so it suffices to verify they have the same roots. It is clear that $0,\dots,j$ are roots. When $r>j$, each summand on the right-hand is non-zero, and the right-hand side can be written as, using the ratio of successive summands,
\[
S_0 \pFq{2}{1}{1+j-p,1+j-s}{2+j-2s}{1},
\]
where $S_0\neq 0$. Hence it suffices to show that ${}_2F_1(1+j-p, 1+j-s; 2+j-2s; 1) = 0$ when $p = 2s-k$ for $1\leq k\leq s-j-1$. This holds by the Chu-Vandermonde identity (see, e.g., \cite[Corollary 2.2.3]{AAR99}), which states, in our case, that
\[
\pFq{2}{1}{1+j-p,1+j-s}{2+j-2s}{1} = \frac{(1+p-2s)(2+p-2s)\cdots(-1+p-s-j)}{(2+j-2s)(3+j-2s)\cdots(-s)}.
\]
The proof of expansion of $r_1(k)$ is now complete. Similarly, starting from
\[
r_2(k) = \gamma \frac{1+\sqrt{1-4x}}{2} + m^2 \frac{1-\sqrt{1-4x}}{2},
\] 
we can deduce as an intermediate step that
\[
B_s =
 \frac{(-1)^s\gamma^pm^s}{(m^2-\gamma)^{2s-1}}\sum_{i=0}^{s-1} (-1)^{i} \binom{p}{i+1}  \frac{i+1}{s}\binom{2s-i-2}{s-1} \gamma^{s-i-1}(m^2-\gamma)^i
\]
and then show it agrees with \eqref{eqn:B_s}. The whole proof is almost identical to that for $r_1(k)$.

The convergence of both series for $0\leq k\leq m$ follows from the absolute convergence of series expansion of $(1+z)^p$ for $|z|\leq 1$. Note that $r_2(m)$ corresponds to $z=-1$.
\end{proof}

We remark that one can continue from $\eqref{eqn:final_series_term}$ to bound that $\sum_s B_s m^s \lesssim 1/m^p$, where the constant depends on $p$. Hence $G_1 + G_2 \simeq 1/m^p$ and thus the gap in \eqref{eqn:gap_expectation} is $\Theta(1/2^m m^p)$ with constant dependent on $p$ only.

\section{Proofs related to Even $p$} \label{sec:even_p_proofs}
Now we prove Lemma~\ref{lem:even_p} below. Since our new $M_{m,k}$ is symmetric, the singular values are the absolute values of the eigenvalues. For $0 < k < m$, $-e_i+e_m$ ($i=k+1,\dots,m-1$) are eigenvectors of eigenvalue $-1$. Hence there are $m-k-1$ singular values of $1$. Observe that the bottom $m-k+1$ rows of $M_{m,k}$ are linearly independent, so the rank of $M_{m,k}$ is $m-k+1$ and there are two more non-zero eigenvalues. Using the trace and Frobenius norm as in the case of the old $M_{m,k}$, we find that the other two eigenvalues $\lambda_1(k)$ and $\lambda_2(k)$ satisfy $\lambda_1(k)+\lambda_2(k) = m-1$ and $\lambda_1^2(k)+\lambda_2^2(k) = (m-1)^2+2k$. Therefore, the singular values $r_{1,2}(k)= |\lambda_{1,2}(k)| = \frac{1}{2}(\sqrt{(m-1)^2+4k}\pm (m-1))$. Formally define $r_{1,2}(k)$ for $k=0$ and $k=m$. When $k=0$, the singular values are actually $r_1(0)$ and $r_2(m)$ and when $k=m$, the singular values are $r_1(m)$ and $r_2(0)$. 
Since $k=0$ and $k=m$ happens with the same probability, this `swap' of singular values does not affect the sum. We can proceed pretending that $r_{1,2}(k)$ are correct for $k=0$ and $k=m$.

Recall that the gap is $\frac{1}{2^{m-1}}(G_1+G_2)$, where $G_1$ and $G_2$ are as defined in \eqref{eqn:G} (we do not need to replace $p/2$ with $p$ here). It remains the same to show that $G_1+G_2\neq 0$ if and only if $m\leq [p/2]$.
\begin{proof}[Proof of Lemma~\ref{lem:even_p}] 
Applying the binomial theorem,
\begin{align*}
 r_1^p(k) + r_2^p(k) &= \frac{1}{2^{p-1}} \sum_{\substack{i: 2|(p-i)}} \binom{p}{i} (m-1)^i ((m-1)^2+ 4k)^{\frac{p-i}{2}}\\
&= \frac{1}{2^{p-1}}  \sum_{\substack{i: 2|(p-i)}}   \binom{p}{i} (m-1)^i
\sum_{j=0}^{\frac{p-i}{2}} \binom{\frac{p-i}{2}}{j} (m-1)^{2j} 4^{\frac{p-i}{2}-j} k^{\frac{p-i}{2}-j}.
\end{align*}
Therefore
\[
G_1 + G_2 = (-1)^m m! \sum_{\substack{i: 2|(p-i)}} \binom{p}{i} (m-1)^i \cdot \sum_{j=0}^{\frac{p-i}{2}} \binom{\frac{p-i}{2}}{j} (m-1)^{2j}4^{\frac{p-i}{2}-j} \stir{\frac{p-i}{2}}{m}.
\]
Note that all terms are of the same sign (interpreting $0$ as any sign) and the sum vanishes only when $\stir{\frac{p-i}{2}}{m} = 0$ for all $i$, that is, when $m > [\frac{p}{2}]$.
\end{proof}

Although when $p$ is even, we have $G_1+G_2=0$, however, we can show that $G_1,G_2\neq 0$, which will be useful for some applications in Section~\ref{sec:app}. It suffices to show the following lemma.
\begin{lemma}\label{lem:single_root} When $p$ is even, the contribution from individual $r_i(k)$ ($i=1,2$) is not zero, provided that $m$ is large enough.
\end{lemma}
\begin{proof}
First we have
\[
r_2^p(k) = \frac{(m-1)^p}{2^p}\sum_{s=0}^\infty \sum_{i=0}^p \binom{p}{i} (-1)^i \binom{i/2}{s} \frac{4^s}{(m-1)^{2s}} k^s.
\]
When $s > p/2$, the binomial coefficient $\binom{i/2}{s}$ vanishes if $i$ is an even integer. Plugging in \eqref{eqn:stirling} we obtain the gap contribution
\[
-\frac{(m-1)^p m!}{2^p} \sum_{s\geq m} \stir{s}{m}\frac{4^s}{(m-1)^{2s}} \sum_{\substack{\text{odd }i\\ 1\leq i\leq p-1}}\binom{p}{i}\binom{i/2}{j}.
\]
Hence it suffices to show that $\sum_s B_s \neq 0$, where
\[
B_s = \stir{s}{m}\frac{4^s}{(m-1)^{2s}} \sum_{\substack{\text{odd }i\\ 1\leq i\leq p-1}}\binom{p}{i}\binom{i/2}{s}.
\]
Note that $B_s$ has alternating signs, so it suffices to show that $|B_{s+1}|<|B_s|$. Indeed,
\[
\frac{\stir{s+1}{m}\frac{4^{s+1}}{(m-1)^{2(s+1)}}}{\stir{s}{m}\frac{4^s}{(m-1)^{2s}}} = \frac{4}{(m-1)^2}\cdot \frac{\stir{s+1}{m}}{
\stir{s}{m}} \leq \frac{8m}{(m-1)^2} < 1
\]
when $m$ is large enough, and
\[
\left|\frac{\binom{i/2}{s+1}}{\binom{i/2}{s}}\right| = \left|\frac{s-\frac{i}{2}}{s+1}\right| < 1.
\]
The proof is now complete.
\end{proof}
It also follows from the proof that for the same large $m$, the gap from $r_i(k)$ has the same sign for all even $p$ up to some $p_0$ depending on $m$. This implies that when $f$ is an even polynomial, the gap contribution from $r_i(k)$ is non-zero.

\section{Algorithm for Even $p$}\label{sec:upperbound}
We first recall the classic result on \text{Count-Sketch} \cite{ccf04}.
\begin{theorem}[\textsc{Count-Sketch}]
There is a randomized linear function $M:\R^n\to \R^S$ with $S = O(w\log(n/\delta))$ and a recovery algorithm $A$ satisfying the following. For any $x\in \R^n$, with probability $\geq 1-\delta$, $A$ reads $Mx$ and outputs $\tilde x\in\R^n$ such that $\|\tilde x-x\|_\infty^2\leq \|x\|_2^2/w$. 
\end{theorem}
We also need a result on $\ell_2$-sampling. We say $x$ is an $(c,\delta)$-approximator to $y$ if $(1-c)y-\delta \leq x\leq (1+c)y+\delta$.
\begin{theorem}[Precision Sampling \cite{ako11}]\label{thm:ell_2_sampling}
Fix $0 < \epsilon < 1/3$. There is a randomized linear function
$M : \R^n \to \R^S$, with $S = O(\epsilon^{-2} \log^3 n)$, and an ``$\ell_p$-sampling algorithm $A$'' satisfying the following. For any non-zero $x \in \R^n$, there is a distribution $D_x$ on $[n]$ such that $D_x(i)$ is an $(\epsilon, 1/\poly(n))$-approximator to $|x_i|^2/\|x\|_2^2$. Then $A$ generates a pair $(i, v)$ such that $i$ is drawn from $D_x$ (using the
randomness of the function $M$ only), and $v$ is a $(\epsilon, 0)$-approximator to $|x_i|^2$.
\end{theorem}
The basic idea is to choose $u_{1},\dots, u_{n}$ with $u_i\sim \Unif(0,1)$ and hash $y_i = x_i/\sqrt{u_i}$ using a \textsc{Count-Sketch} structure of size $\Theta(w\log n)$ (where $w=\Theta(\epsilon^{-1}\log n + \epsilon^{-2})$), and recover the heaviest $y_i$ and thus $x_i$ if $y_i$ is the unique entry satisfying $y_i \geq C\|x\|_2^2/\epsilon$ for some absolute constant $C$, which happens with the desired probability $|x_i|^2/\|x\|_2^2 \pm 1/\poly(n)$. The estimate error of $x_i$ follows from \textsc{Count-Sketch} guarantee.

Now we turn to our algorithm. Let $A=(a_{ij})$ be an integer matrix and suppose that the rows of $A$ are $a_1,a_2,\dots$. There are $O(1)$ non-zero entries in each row and each column. Assume $p\geq 4$. We shall use the structure for $\ell_2$ sampling on $n$ rows while using a bigger underlying \textsc{Count-Sketch} structure to hash all $n^2$ elements of a matrix. 

For simplicity, we present our algorithm in Algorithm~\ref{alg:even_p} with the assumption that $u_1,\dots,u_n$ are i.i.d.\ $\Unif(0,1)$. The randomness can be reduced using the same technique in \cite{ako11} which uses $O(\log n)$ seeds.

\begin{algorithm}[tbh]
\vspace{2mm}
\caption{Algorithm for even $p$ and sparse matrices}\label{alg:even_p}
\begin{algorithmic}[1]
\Statex Assume that matrix $A$ has at most $k = O(1)$ non-zero entries per row and per column.
\State $T \gets \Theta(n^{1-2/p}/\epsilon^{2})$
\State $R \gets \Theta(\log n)$
\State $w \gets O(\epsilon^{-1}\log n + \epsilon^{-2})$
\State $I_s \gets \emptyset$ is a multiset for $s=1,\dots,p/2$
\State Choose i.i.d.\ $u_{1},\dots, u_{n}$ with $u_i\sim \Unif(0,1)$.
\State $D \gets \diag\{1/\sqrt{u_{1}},\dots,1/\sqrt{u_{n}}\}$
\State In parallel, maintain $p/2$ \textsc{Count-Sketch} structures $\mathcal{S}_s$ ($s\in[p]$) of size $\Theta(\epsilon^{-1}T\log n)$.
\State Maintain a sketch for estimating $\|A\|_F^2$ and obtain a $(1+\epsilon)$-approximation $L$ as in \cite{ams99}
\State In parallel, maintain $pT/2$ structures $\mathcal{P}_{s,t}$ ($(s,t)\in [p/2]\times [T]$), each has $R$ repetitions of the Precision Sampling structure for all $n^2$ entries of $B = DA$, $t=1,\dots, T$. The Precision Sampling structure uses a \textsc{Count-Sketch} structure of size $O(w\log n)$.
\State Maintain a sketch for estimating $\|B\|_F^2$ and obtain an $(1+\epsilon)$-approximation $L'$ as in \cite{ams99}
\For{$s\gets 1 \textrm{ to } p/2$}
\For{$t\gets 1 \textrm{ to } T$}
\For{$r\gets 1 \textrm{ to } R$}
\State Use the $r$-th repetition of the $\mathcal{P}_{s,t}$ to obtain estimates $\tilde b_{i'1},\dots, \tilde b_{i'n}$ for all $i'$ and form rows $\tilde b_{i'} = (b_{i'1},\dots,b_{i'n})$.
\State If there exists a unique $i'$ such that $\|\tilde b_i'\|_2^2 \geq C'L/\epsilon$ for some appropriate absolute constant $C'$, return $i'$ and exit the inner loop 
\EndFor
\State Retain only entries of $b_{i'}$ that are at least $2L'/\sqrt{w}$.\label{alg:step:threshold}
\State $\tilde a_{i'} \gets \sqrt{u_{i'}}\tilde b_{i'}$
\State $I_s \gets I_s \cup \{i'\}$
\EndFor
\EndFor
\For{$s\gets 1 \textrm{ to } p/2$}
\State Use $\mathcal{S}_s$ to obtain estimates $\tilde a'_{i'1},\dots, \tilde a'_{i'n}$ for all $i'$ and form rows $\tilde a'_{i'} = (a'_{i'1},\dots,a'_{i'n})$.
\State Find all $i$ such that $\|\tilde a'_{i'}\|_2^2\geq L/(10T)$ and retain $O(T)$ of them corresponding to the largest $\|\tilde a'_{i'}\|_2^2$, making a set $K_s$
\State $\tilde a_i \gets \tilde a_i'$ for all $i\in K_s$
\State $I_s = I_s \cup K_s$
\EndFor
\State Return $Y$ as defined in \eqref{eqn:estimator}
\end{algorithmic}
\end{algorithm}

\begin{theorem} For sparse matrices $A$ with $O(1)$ non-zero entries per row and per column, Algorithm~\ref{alg:even_p} returns a value that is a $(1+\epsilon)$-approximation to $\|A\|_p^p$ with constant probability, using space $O(n^{1-2/p}\poly(1/\epsilon,\log n))$, where the constant in the $O$-notation depends on $p$.
\end{theorem}
\begin{proof}
It is the guarantee from the underlying \textsc{Count-Sketch} structure of size $\Theta(w\log n)$ (where $w=O(\epsilon^{-1}\log n+\epsilon^{-2})$) that
\[
\tilde b_{i'j} = b_{i'j} \pm \sqrt{\frac{\|B\|_F^2}{w}}
\]
for all $j$. Since there are only $O(1)$ non-zero entries in $b_{i'}$, we can use a constant-factor larger size $w' = O(w)$ for \textsc{Count-Sketch} such that 
\[
\tilde b_{i'j} = b_{i'j} \pm \sqrt{\frac{\|B\|_F^2}{w'}}
\]
and thus
\begin{equation}\label{eqn:CS_error_row_norm}
\|\tilde b_{i'}\|_2^2 = \|b_{i'}\|_2^2 \pm \frac{\|B\|_F^2}{w}.
\end{equation}
Since each row $i$ is scaled by the same factor $1/\sqrt{u_i}$, we can apply the proof of Theorem~\ref{thm:ell_2_sampling} to the vector of row norms $\{\|a_i\|_2\}$ and $\{\|b_i\|_2\}$, which remains still valid because of the error guarantee \eqref{eqn:CS_error_row_norm} which is analogous to the $1$-dimensional case. It follows that with probability $\geq 1-1/n$ (since there are $\Theta(\log n)$ repetitions in each of the $T$ structures), an $i'$ is returned from the inner for-loop such that
\begin{equation}\label{eqn:ell_2_sampling_original}
\Pr\{i' = i\} = (1\pm \epsilon)\frac{\|a_i\|_2^2}{\|A\|_F^2} \pm \frac{1}{\poly(n)}.
\end{equation}
Next we analyse estimation error. It holds with high probability that $\|B\|_F^2\leq w\|A\|_F^2$. Since $a_i$ (and thus $b_i$) has $O(1)$-elements, the heaviest element $a_{i'j'}$ (resp. $b_{i'j'})$ has weight at least a constant fraction of $\|a_i\|_2$ (resp. $\|b_i\|_2$). It follows from the thresholding condition of the returned $\|b_i\|_2$ that we can use a constant big enough for $w' = O(w)$ to obtain
\[
\tilde a_{i'j'} = \sqrt{u_{i'}}\cdot \tilde b_{i'j'} = (1\pm \epsilon) a_{i'j'},
\]
Suppose that the heaviest element is $b_{ij}$. Similarly, if $|a_{i\ell}| \geq \eta|a_{ij}|$ (where $\eta$ is a small constant to be determined later), making $w' = \Omega(w/\eta)$, we can recover 
\[
\tilde a_{i\ell} = \sqrt{u_i}\cdot \tilde b_{i\ell} = a_{i\ell} \pm \epsilon\eta a_{ij} = (1\pm \epsilon)a_{i\ell}.
\]
Note that there are $O(1)$ non-zero entries $a_{i\ell}$ such that $|a_{i\ell}|\leq \eta|a_{ij}|$ and each of them has at most $\Theta(\epsilon\eta a_{ij})$ additive error by the threshold in Step~\ref{alg:step:threshold}, the approximation $\tilde a_i$ to $a_i$ therefore satisfies
\[
\|\tilde a_i - a_i\|_2^2 \leq \epsilon^2 \|a_i\|_2^2 + O(1) \cdot \epsilon^2 \eta^2 \|a_i\|_2^2 \leq 2\epsilon^2 \|a_i\|_2^2
\]
by choosing an $\eta$ small enough. It follows that $\|\tilde a_i\|_2$ is a $(1+\Theta(\epsilon))$-approximation to $\|a_i\|_2$, and $|\langle\tilde a_i,\tilde a_j\rangle|= |\langle a_i,a_j\rangle| \pm \Theta(\epsilon)\|a_i\|_2 \|a_j\|_2$. 

Similarly by a standard heavy hitter argument, with probability $\Omega(1)$, the set $K_s$ contains all $i$ such that $\|a_{i}\|_2^2\geq \epsilon\|A\|_F^2/T$, if we choose size of the \text{Count-Sketch} structure with a large enough heading constant. This implies that if $i\in I_s\setminus K_s$, then $\|a_{i}\|_2^2\leq \epsilon\|A\|_F^2/T$.

Next we show that our estimate is desirable. First, we observe that the additive $1/\poly(n)$ term in \eqref{eqn:ell_2_sampling_original} can be dropped at the cost of increasing the total failure probability by $1/\poly(n)$. Hence we may assume in our analysis that
\begin{equation}\label{eqn:ell_2_sampling}
\Pr\{i' = i\} = (1\pm \epsilon)\frac{\|a_i\|_2^2}{\|A\|_F^2}.
\end{equation}
For notational simplicity let $q=p/2$ and $\ell_{i} = \|\tilde a_{i}\|_2^2$ if $i$ is a sampled row. For $i_s\in I_s$, define
\[
\tau(i_s) = \begin{cases}
			  1,            & i_s \in K_s;\\
			  \|A\|_F^2/\|a_{i_s}\|_2^2, & \text{otherwise}.
		   \end{cases},\qquad
\tilde \tau(i_s) = \begin{cases}
			  1,            & i_s \in K_s;\\
			  L/\ell_{i_s}, & \text{otherwise}.
		   \end{cases}	
\]
and for $(i_1,\dots,i_q)\in I_1\times \cdots \times I_q$, define
\begin{gather*}
\tau(i_1,\dots,i_q) = \tau(i_1)\cdots\tau(i_q)\\
\tilde\tau(i_1,\dots,i_q) = \tilde\tau(i_1)\cdots\tilde\tau(i_q)
\end{gather*}
and
\begin{gather*}
X(i_1,\dots,i_q) = \prod_{j=1}^q \langle a_{i_i}, a_{i_{j+1}}\rangle \tau(i_1)\cdots\tau(i_q)\\
\tilde X(i_1,\dots,i_q) = \prod_{j=1}^q \langle \tilde a_{i_i}, \tilde a_{i_{j+1}}\rangle \tilde\tau(i_1)\cdots\tilde\tau(i_q),
\end{gather*}
where it is understood that $a_{i_{q+1}}=a_{i_1}$. Also let
\[
p_{s,t}(i) = \Pr\{\text{row }i\text{ gets sampled in }(s,t)\text{-th precision sampling}\}.
\]
We claim that
\[
\|A\|_p^p = \sum_{1\leq i_1,\dots,i_q\leq n} \prod_{j=1}^q \langle a_{i_j},a_{i_{j+1}}\rangle.
\]
When $q = p/2$ is odd,
\begin{align*}
 \|A\|_p^p &= \|\underbrace{(A^TA)\cdots (A^TA)}_{(q-1)/2\text{ times}}A^T\|_F^2 \\
&= \sum_{k, \ell} \sum_{i_1,\dots,i_{q-1}} (A^T_{k,i_1} A_{i_1,i_2} A^T_{i_2, i_3} A_{i_3, i_4} \cdots A_{i_{q-2},i_{q-1}} A^T_{i_{q-1},\ell})^2 \\
&= \sum_{k, \ell} \sum_{\substack{i_1,\dots,i_{q-1}\\j_1,\dots,j_{q-1}}} A_{i_1,k}A_{j_1,k}A_{i_1,i_2}A_{j_1,j_2}\cdots A_{\ell,i_{q-1}}A_{\ell,j_{q-1}}\\
&= \sum \langle a_{i_1},a_{j_1}\rangle \!\cdot\!\!\!\!\!\!\! \prod_{\substack{\text{odd }t\\ 1\leq t\leq q-2}} \!\!\!\!\! \langle a_{i_t},a_{i_{t+2}}\rangle \langle a_{j_t},a_{j_{t+2}}\rangle\!\cdot\! \langle a_{i_{q-2}},a_\ell\rangle \langle a_{j_{q-2}},a_\ell\rangle,
\end{align*}
which is a `cyclic' form of inner products and the rightmost sum is taken over all appearing variables ($i_t$, $j_t$ and $\ell$) in the expression. A similar argument works when $q$ is even. 

Our estimator is
\begin{equation}\label{eqn:estimator}
Y = \sum_{i_1\in I_1,\dots,i_q\in I_q} \frac{1}{T^{\sigma(i_1,\dots,i_q)}}\tilde X(i_1,\dots,i_q),
\end{equation}
where
\begin{equation*}
\sigma(i_1,\dots,i_q) = |\{s: i_s\not\in K_s\}|.
\end{equation*}
Then
\[
\left|\E Y - \|A\|_p^p\right| \leq 
\sum_{i_1,\dots,i_q} \left|\frac{\Pr\{i_1\in I_1\}\cdots \Pr\{i_q\in I_q\}}{T^\sigma}\tilde X(i_1,\dots,i_q) - \prod_{j=1}^q \langle a_{i_i}, a_{i_{j+1}}\rangle\right|.
\]
Observe that $i\in I_s\setminus K_s$,
\[
\Pr\{i \in I_s\} = 1-\prod_{t=1}^T (1-p_{s,t}(i)) = (1\pm O(\epsilon))\frac{T}{\tau(i)},
\]
where we used the fact that $p_s(i_s)=(1\pm\epsilon)/\tau(i)$ and $1/\tau(i)\leq \epsilon/T$ for $i\in I_s\setminus K_s$. For $i_s \in K_s$, we have
\[
\Pr\{i \in I_s\} = \frac{1}{\tau(i)} = 1.
\]
Hence
\[
\frac{\Pr\{i_1\in I_1\}\cdots\Pr\{i_q\in I_q\}}{T^\sigma} = (1\pm O(\epsilon))\frac{1}{\tau(i_1)\cdots\tau(i_s)} = (1\pm O(\epsilon))\frac{1}{\tilde\tau(i_1)\cdots\tilde\tau(i_s)}
\]
and
\[
\frac{\Pr\{i_1\in I_1\}\cdots\Pr\{i_q\in I_q\}}{T^\sigma}\tilde X(i_1,\dots,i_q) =  (1\pm O(\epsilon))\prod_{j=1}^q \langle \tilde a_{i_j},\tilde a_{i_{j+1}}\rangle.
\]
it follows that
\[
\left|\E Y - \|A\|_p^p\right| \leq 
\sum_{i_1,\dots,i_q} \left\{\left|\prod_{j=1}^q \langle \tilde a_{i_j},\tilde a_{i_{j+1}}\rangle - \prod_{j=1}^q \langle a_{i_i}, a_{i_{j+1}}\rangle\right| + O(\epsilon) \left|\prod_{j=1}^q \langle \tilde a_{i_j},\tilde a_{i_{j+1}}\rangle\right|\right\}.
\]
The key observation is that each $a_i$ has only $O(1)$ rows with overlapping support, since each row and each column has only $O(1)$ non-zero entries. The same claim holds for $\tilde a_i$, which is due to our threshold in Step~\ref{alg:step:threshold}: for an entry to be retained, it must be larger than $\|B\|_F/\sqrt{w}$ (the uniform additive error from \textsc{Count-Sketch}), which is impossible for zero entries. Therefore, each row $i$ appears in $O(1)$ contributing summands. Each contributing summand is bounded by
\[
\Theta(1)\cdot \epsilon \prod_{j=1}^q \|a_{i_j}\|_2^2 \leq \Theta(1) \cdot \epsilon\max\{\|a_{i_1}\|_2^{2q},\dots,\|a_{i_q}\|_2^{2q}\}.
\]
Therefore
\begin{equation}\label{eqn:tildeX_expectation}
\left|\E Y - \|A\|_p^p\right| \lesssim \epsilon \sum_i \|a_i\|_2^{2q} \leq \epsilon\|A\|_{2q}^{2q}.
\end{equation}
as desired, where the last inequality follows from the fact of Schatten $r$-norms ($r\geq 1$) that $\|M\|_r^r \geq \sum_{i=1}^n |M_{ii}|^r$ and choosing $M = A^TA$ and $r=q$.  

Next we bound the variance. 
\[
\E Y^2 = \E\sum_{\substack{i_1\in I_1,\dots,i_q\in I_q\\j_1\in I_1,\dots,j_q\in I_q}} \frac{1}{T^{\sigma(i_1,\dots,i_q)}T^{\sigma(j_1,\dots,j_q)}} \tilde X(i_1,\dots,i_q)\tilde X(j_1,\dots,j_q).
\]
Similarly to before, the right-hand side can be simplified as 
\begin{equation}\label{eqn:variance_expansion}
\sum_{r=0}^q \sum_{\substack{i_1\in I_1,\dots,i_q\in I_q\\j_1\in I_1,\dots,j_q\in I_q\\ |\{s: i_s=j_s\not\in K_s\}|=r}} (1+O(\epsilon))\frac{\prod_{s: i_s=j_s\not\in K_s} \tau(i_s)}{T^r} \prod_{s=1}^q \langle \tilde a_{i_s},\tilde a_{i_{s+1}}\rangle \prod_{s=1}^q \langle \tilde a_{j_s},\tilde a_{j_{s+1}}\rangle.
\end{equation}
We can upper bound each individual summand as
\begin{align*}
&\quad\ \left|(1+O(\epsilon))\frac{\prod_{s: i_s=j_s\not\in K_s} \tau(i_s)}{T^r} \prod_{s=1}^q \langle \tilde a_{i_s},\tilde a_{i_{s+1}}\rangle \prod_{s=1}^q \langle \tilde a_{j_s},\tilde a_{j_{s+1}}\rangle\right|\\
 &\lesssim \frac{1}{T^r}\frac{\|A\|_F^{2r}}{\prod_{s: i_s=j_s\not\in K_s} \|a_{i_s}\|_2^2}\prod_{s=1}^q \|a_{i_s}\|_2^2 \cdot \prod_{s=1}^q \|a_{j_s}\|_2^2\\
&\leq \frac{1}{T^r}\|A\|_F^{2r} \left(\max_i \|a_i\|_2^2\right)^{2q-r}.
\end{align*}
Now, note that the terms corresponding to $r=0$ in \eqref{eqn:variance_expansion} are covered by the expansion of $(\E Y)^2$. And by the same argument as before, each $i_s$ or $j_s$ appears in $O(1)$ contributing summands, we have that
\[
\E Y^2 - (\E Y)^2 \lesssim \sum_{r=1}^q \frac{1}{T^r}\|A\|_F^{2r} \|A\|_p^{2p-2r} \leq \sum_{r=1}^q \frac{1}{T^r}n^{r(1-\frac2p)} \|A\|_p^{2p},
\]
which implies that
\[
\E Y^2 - (\E Y)^2\leq \epsilon^2 \|A\|_{p}^{2p}
\]
if the constant $C$ in $T=Cn^{1-2/p}/\epsilon^2$ is large enough.

\end{proof}

\section{General Functions and Applications}\label{sec:app}

The following is a direct corollary of Theorem~\ref{thm:schatten}. 
\begin{theorem}\label{thm:general}
Let $f$ be a diagonally block-additive function. Suppose that $f(x)\simeq x^p$ for $x$ near 0 or $x$ near infinity, where $p > 0$ is not an even integer. For any even integer $t$, there exists a constant $c=c(t)>0$ such that any streaming algorithm that approximates $f(X)$ within a factor $1\pm c$ with constant error probability must use $\Omega_t(N^{1-1/t})$ bits of space.
\end{theorem}
\begin{proof}
Suppose that $f(x)\sim \alpha x^p$ for $x$ near $0$, that is, for any $\eta > 0$, there exists $\delta = \delta(\eta) > 0$ such that $\alpha(1-\eta) f(x)\leq x^p \leq \alpha(1+\eta) f(x)$ for all $x \in [0,\delta)$.

Let $c_0$ be the approximation ratio parameter in Theorem~\ref{thm:schatten} for Schatten $p$-norm. Let $\epsilon$ be sufficiently small (it could depend on $t$ and thus $m$) such that the singular values of $\eps\mathcal{M}$ are at most $\delta(c_0/3)$, where $\mathcal{M}$ is the hard instance matrix used in Theorem~\ref{thm:schatten}. Then $\alpha(1-c_0/3)  f(\eps\mathcal{M}) \leq \|\eps\mathcal{M}\|_p^p\leq \alpha(1+c_0/3)f(\eps\mathcal{M})$. Therefore, any algorithm that approximates $f(\eps\mathcal{M})$ within a factor of $(1\pm c_0/3)$ can produce a $(1\pm c_0)$-approximation of $\|\epsilon\mathcal{M}\|_p^p$. The lower bound follows from Theorem~\ref{thm:schatten}.

When $f(x)\simeq x^p$ for $x$ near infinity, a similar argument works for $\lambda\mathcal{M}$ where $\lambda$ is sufficiently large.
\end{proof}

The following is a corollary of Lemma~\ref{lem:even_p}.
\begin{theorem}\label{thm:general_2}
Suppose that $f$ admits a Taylor expansion near $0$ that has infinitely many even-order terms of non-zero coefficient. Then for any arbitrary large $m$, there exists $c = c(m)$ such that any data stream algorithm which outputs, with constant error probability, a $(1+c)$-approximation to $\|X\|_p^p$ requires $\Omega(N^{1-1/m})$ bits of space.
\end{theorem}
\begin{proof}
If the expansion has a odd-order term with non-zero coefficient, apply Theorem~\ref{thm:general} with the lowest non-zero odd-order term. Hence we may assume that all terms are of even order.
For any given $m$, there exists $p > 2m$ such that the $x^p$ term in the Taylor expansion of $f$ has a non-zero coefficient $a_p$. Let $p$ be the lowest order of such a term, and write
\[
f(x) = \sum_{i=0}^{p-1} a_i x^{p-1} + a_p x^p + O(x^{p+1}).
\]
Let $\epsilon > 0$ be a small constant to be determined later and consider the matrix $\epsilon\mathcal{M}$, where $\mathcal{M}$ is our hard instance matrix used in Lemma~\ref{lem:even_p}. Lemma~\ref{lem:even_p} guarantees a gap of $f(\epsilon\mathcal{M})$, which is then $a_p \epsilon^p G + R(\epsilon)$, where $G$ is the gap for $x^p$ on unscaled hard instance $\mathcal{M}$ and $|R(\epsilon)|\leq K\epsilon^{p+1}$ for some constant $K$ depending only on $f(x)$, $m$ and $p$. Choosing $\epsilon < a_p G/K$ guarantees that the gap $a_p\epsilon^p G + R(\epsilon)\neq 0$.
\end{proof}

Now we are ready to prove the lower bound for some eigenvalue shrinkers and $M$-estimators. The following are the three optimal eigenvalue shrinkers from \cite{GD14}:
\begin{gather*}
\eta_1(x) = \begin{cases}
				\frac{1}{x}\sqrt{(x^2-\alpha-1)^2-4\alpha}, & x\geq 1+\sqrt{\alpha};\\
				0,  &x<1+\sqrt{\alpha},
			\end{cases}\\
\eta_2(x) = \begin{cases}
			\frac{1}{\sqrt 2}\sqrt{x^2-\alpha-1+\sqrt{(x^2\!-\!\alpha\!-\!1)^2-4\alpha}}, &x\geq 1+\sqrt{\alpha};\\
           0, & x<1+\sqrt{\alpha},
           \end{cases}\\
\eta_3(x) = \frac{1}{x\eta_2^2(x)} \max \left\{\eta_2^4(x)-\alpha-\alpha x\eta_2(x), 0\right\}\end{gather*}
where we assume that $0\cdot\infty = 0$. Since $\eta_i(x) \simeq x$ when $x$ is large, the lower bound follows from Theorem~\ref{thm:general}.

Some commonly used influence functions $\rho(x)$ can be found in \cite{Zhang95}, summarized in Table~\ref{tab:M_estimators}. Several of them are asymptotically linear when $x$ is large and Theorem~\ref{thm:general} applies. Some are covered by Theorem~\ref{thm:general_2}. For the last function, notice that it is a constant on $[c,+\infty)$, we can rescale our hard instance matrix $\mathcal{M}$ such that the larger root $r_1(k)$ falls in $[c,+\infty)$ and the smaller root $r_2(k)$ in $[0,c]$. The larger root $r_1(k)$ therefore has no contribution to the gap. The contribution from the smaller root $r_2(k)$ is nonzero by the remark following Lemma~\ref{lem:single_root}.

\begin{table*}
\vspace{2mm}
\centering
\begin{tabular}{|l|c|l|c|}
\hline
Function $\rho(x)$ & Apply & Function $\rho(x)$ & Apply \\
\hline
$2(\sqrt{1+x^2/2}-1)$ & Theorem~\ref{thm:general} & $\frac{x^2/2}{1+x^2}$ & Theorem~\ref{thm:general} \\
\hline
$c^2(\frac{x}{c}-\ln(1+\frac{x}{c}))$ & Theorem~\ref{thm:general} & $\frac{c^2}{2}(1-\exp(-x^2/c^2))$ & Theorem~\ref{thm:general_2} \\
\hline
$\begin{cases} 
x^2/2, & x\leq k;\\
k(x-k/2), & x > k
\end{cases}$                              & Theorem~\ref{thm:general} &
$\begin{cases}
\frac{c^2}{6}(1-(1-x^2/c^2)^3), & x \leq c;\\
c^2/6, & x > c
\end{cases}$                              & $\begin{array}{c} \text{Remark after}\\ \text{Lemma~\ref{lem:single_root}}\end{array}$ \\
\hline
$\frac{c^2}{2}\ln(1+\frac{x^2}{c^2})$     & Theorem~\ref{thm:general_2} &  &   \\
\hline
\end{tabular}
\caption{Application of Theorem~\ref{thm:general} and Theorem~\ref{thm:general_2} to some M-estimators from \cite{Zhang95}.}
\label{tab:M_estimators}
\end{table*}

Finally we consider functions of the form \[
F_k(X) = \sum_{i=1}^k f(\sigma_i(X))
\] and prove (a slightly rephrased) Theorem~\ref{thm:partial_singular_values} in the introduction. 
\begin{theorem} Let $\alpha \in (0,1/2)$. Suppose that $f$ is strictly increasing. There exists $N_0$ and $c_0$ such that for all $N\geq N_0$, $k\leq \alpha N$ and $c\in (0,c_0)$, any data stream algorithm which outputs, with constant error probability, a $(1+c)$-approximation to $F_k(X)$ of $X\in \R^{N\times N}$ requires\\ $\Omega_{\alpha}(N^{1+\Theta(1/\ln\alpha)})$ bits of space.
\end{theorem}

\begin{proof}
Similarly to Theorem~\ref{thm:gap_expectation} we reduce the problem from the $\textsc{BHH}_n^0$ problem. Let $m=t$ be the largest integer such that $1/(t2^t)\geq\alpha$. Then $m=t=\Theta(\ln(1/\alpha))$. We analyse the largest $k$ singular values of $\mathcal{M}$ as defined in \eqref{eqn:tutte_matrix_big}. Recall that $q_1,\dots,q_{n/m}$ are divided into $N/(2m)$ groups. Let $X_1,\dots,X_{N/(2m)}$ be the larger $q_i$'s in each group, then $X_1,\dots,X_{N/(2m)}$ are i.i.d.\ random variables. In the even case, they are defined on $\{m/2,m/2+2,\dots,m\}$ subject to the distribution
\[
\Pr\left\{X_1 = \frac{m}{2} +  j\right\} = \begin{cases}
													p_m(\frac{m}{2}), & j = 0;\\
													2p_m(\frac{m}{2}+j), & j > 0,
												\end{cases}
											\quad j = 0, 2, \dots, \frac{m}{2}.
\]
In the odd case, they are defined on $\{m/2+1,m/2+3,\dots,m-1\}$ with probability density function
\[
\Pr\left\{X_1 = \frac{m}{2} + j\right\} = 2p_m\left(\frac{m}{2} + j\right), \quad j = 1,3,\dots,\frac{m}{2}-1.
\]
With probability $1/2^{m-2}$, $X_i = m$ in the even case and with probability $m/2^{m-2}$, $X_i = m/2-1$ in the odd case. It immediately follows from a Chernoff bound that with high probability, it holds that $X_i = m$ (resp.\ $X_i=m-1$) for at least $(N/2m)\cdot(1/2^{m-2})(1-\delta) = (1-\delta)N/(m2^{m-1})$ different $i$'s in the even case (resp.\ odd case). Since $r_1(m-1) < r_1(m)$ and $f$ is strictly increasing, the value $F_k(X)$, when $k\leq \alpha N\leq (1-\delta)N/(m2^{m-1})$, with high probability, exhibits a gap of size at least $c\cdot k$ for some constant $c$ between the even and the odd cases. Since $F_k(\mathcal{M}) = \Theta(k)$ with high probability, the lower bound for Ky-Fan $k$-norm follows from the lower bound for $\textsc{BHH}_n^0$.
\end{proof}

The lower bound for Ky-Fan $k$-norms follows immediately. For $k\leq \alpha N$ it follows from the preceding theorem with $f(x)=x$; for $k > \alpha N$, the lower bound follows from our lower bound for the Schatten $1$-norm by embedding the hard instance of dimension $\alpha N\times \alpha N$ into the $N\times N$ matrix $X$, padded with zeros.

As the final result of the paper, we show an $\Omega(n^{1-1/t})$ lower bound for SVD entropy function of matrices, in the following subsection.

\subsection{SVD Entropy}\label{sec:entropy}
Let $h(x) = x^2\ln x^2$. For $X\in \R^{N\times N}$, we define its SVD entropy $H(X)$ as 
\[
H(X) = \sum_i h\left(\frac{\sigma_i(X)}{\|X\|_F}\right)
\]
In this section, our goal is to show the following theorem.
\begin{theorem}\label{cor:original_entropy}
Let $t$ be an even integer and $X\in \R^{N\times N}$, where $N$ is sufficiently large. There exists $c=c(t)$ such that estimating the matrix entropy $H(X)$ within an additive error of $c$ requires $\Omega_t(N^{1-1/t})$ bits of space.
\end{theorem}
This theorem will follow easily from the following lemma, whose proof is postponed to later in the section. It is based on Theorem~\ref{thm:gap_expectation} with the same hard instance used by Theorem~\ref{thm:schatten}.
\begin{lemma}\label{lem:small_entropy}
Let $t$ be an even integer and $X\in \R^{N\times N}$, where $N$ is sufficiently large. There exists a constant $c = c(t) > 0$ such that any algorithm that approximates $h(X)$ within a factor $1\pm c$ with constant probability in the streaming model must use $\Omega_t(N^{1-1/t})$ bits of space.
\end{lemma}
\begin{proof}[Proof of Theorem~\ref{cor:original_entropy}]
Let $X$ be the matrix in the hard instance for estimating $h(X)$, which is the same hard instance used by Theorem~\ref{thm:schatten}. Then $X$ consists of smaller diagonal blocks of size $m = m(t)$, and $\|X\|_F^2 = CN$, where $C = C(m,t)$ is a constant depending on $t$ and $m$ only. It is also easy to see that $K_1N\leq h(X)\leq K_2N$ for some constants $K_1$, $K_2$ depending only on $m$ and $t$.

Now we show that an additive $c$-approximation to $H(X)$ can yield a multiplicative $(1+c')$-approximation to $H(X)$. Recall that $H(X) = \ln\|X\|_F^2 - h(X) / \|X\|_F^2 = \ln(CN)- h(X)/(CN)$. Suppose that $Z$ is an additive $c$-approximation to $H(X)$, then we compute $\hat X = CN\ln (CN) - CN Z$. Since $Z\leq \mathcal{H}(Y) + c$,
\[
\hat X \geq CN \ln(CN) - CN(h(X) + c) = h(X) - cCN \geq \left(1 - \frac{cC}{K_1}\right) h(Y).
\]
Similarly it can be shown that $\hat X\leq (1 + \frac{cC}{K_2})h(X)$, thus choosing $c' = Cc/\max\{K_1,K_2\}$ suffices.

The lower bound follows from Lemma~\ref{lem:small_entropy}.
\end{proof}

We devote the rest of the section to the proof of Lemma~\ref{lem:small_entropy}, for which we shall apply Theorem~\ref{thm:gap_expectation} to $h(x)$.
\begin{proof}[Proof of Lemma~\ref{lem:small_entropy}]
Following the same argument as in the proof of Theorem~\ref{thm:schatten}, our goal is to show that
\[
G_1 + G_2\neq 0,
\]
where
\[
G_i = \sum_k (-1)^k\binom{m}{k} h(\sqrt{r_i(k)}) = \sum_k (-1)^k\binom{m}{k} r_1(k)\ln r_1(k),\quad i=1,2.
\]
Taking $\gamma=1$ in the definition of $M_{m,k}$ in \eqref{eqn:asymmetric_M}, we obtain (see Section~\ref{sec:power_series}) that
\begin{gather}
r_1(k) = m^2 + \sum_{j=1}^\infty (-1)^{j-1} \frac{C_{j-1} m^j k^j}{(m^2-1)^{2j-1}} \label{eqn:r_1(k)}\\
r_2(k) = 1 + \sum_{j=1}^\infty (-1)^j \frac{C_{j-1} m^j k^j}{(m^2-1)^{2j-1}}, \label{eqn:r_2(k)}
\end{gather}
where $C_j$ denotes the $j$-th Catalan number. 
Plugging \eqref{eqn:r_2(k)} into
\[
(1+x)\ln (1+x) = x + \sum_{n\geq 2} (-1)^n\frac{x^n}{n(n-1)},\quad |x|\leq 1
\]
and arranging the terms as in Section~\ref{sec:power_series} yields that
\[
r_2(k)\ln r_2(k) = \sum_s B_s k^s,
\]
where
\[
B_s = \frac{(-1)^s m^s}{s(m^2-1)^{2s}}\left( (m^2-1)\binom{2s-2}{s-1} + \sum_{i=2}^s \frac{(m^2-1)^i (-1)^i}{(i-1)}\binom{2s-i-1}{s-1} \right), \quad s\geq 2.
\]
Let
\[
D_i = \frac{x^i}{i-1}\binom{2s-i-1}{s-1},\] 
then
\[
\frac{D_{i+1}}{D_i} = \frac{(i-1)(s-i)}{i(2s-i-1)}x.
\]
Let $i^\ast = \max_i D_i$. One can obtain by solving $D_{i+1}/D_i\geq 1$ that $i^\ast = \lceil \frac{x-2}{x-1}s\rceil$. Hence $i^\ast = s$ for $s\leq m^2-3$, and $\sum (-1)^iD_i \simeq (-1)^sD_s$ when $s < \alpha m^2$ for some $\alpha \in (0,1)$. Note that $B_s$ has the same sign for $s \leq 
\alpha m^2$ because $(m^2-1)\binom{2s-2}{s-1}$ is negligible compared with $D_s$. The partial sum (choosing even $m$)
\[
\sum_{s=m}^{\alpha m^2} B_s \stir{s}{m}m! 
\gtrsim \sum_{s=m}^{\alpha m^2} \frac{m^s}{s(m^2-1)^{2s}}\cdot \frac{(m^2-1)^s}{s-1} \stir{s}{m}m!
\gtrsim \frac{1}{m^2 e^m}.
\]
Next we show that $G_2\gtrsim 1/m^{m+2}$.

Write
\[
\sum_{i=2}^s \frac{(m^2-1)^i (-1)^i}{i-1}\binom{2s-i-1}{s-1} = (m^2-1)^2\binom{2s-3}{s-1}\pFq{3}{2}{1,1,2-s}{2,3-2s}{1-m^2}.
\]
We can write (see, e.g., \cite{Wolfram1})
\[
\pFq{3}{2}{1,1,2-s}{2,3-2s}{1-m^2} = \frac{1}{m^2-1} \int_{1-m^2}^0 \pFq{2}{1}{1,2-s}{3-2s}{x} dx.
\]
Arranging the terms, we can write
\[
B_s = \frac{(-1)^s m^s}{s(m^2-1)^{2s}}(m^2-1)\binom{2s-2}{s-1}B_s',
\]
where
\[
B_s' = 1 + \frac{1}{2}\int_{1-m^2}^0 \pFq{2}{1}{1,2-s}{3-2s}{x} dx.
\]

It is shown in \cite[Theorem 3.1]{KK02} that ${}_2F_1(1,2-s;3-2s;x)$ has no real root on $(-\infty,0]$ when $s$ is even and has a single root on  $(-\infty,0]$ when $s$ is odd. Therefore, $B_s'>0$ when $s$ is even and thus $B_s > 0$. Note that (see, e.g., (2.5.1) in \cite{AAR99})
\[
\frac{d}{dx}{}_2F_1(1,2-s;3-2s;x) = \frac{s-2}{2s-3}{}_2F_1(2,3-s;4-2s;x).
\]
Again applying \cite[Theorem 3.1]{KK02} gives that ${}_2F_1(2,3-s;4-2s;x) > 0$ on $(-\infty, 0]$ when $s$ is odd. Hence ${}_2F_1(1,2-s;3-2s;x)$ is increasing on $(-\infty, 0]$ when $s$ is odd, and
\[
B_s' \leq 1 + \frac{1}{2}\int_{1-m^2}^0 1\ dx = \frac{m^2+1}{2}.
\]
Let $I = \{\text{odd }s: B_s' > 0\}$. From the argument on the previous page we know that $s> \alpha m^2$ for any $s\in I$. It follows that
\[
\left|\sum_{s\in I} B_s \stir{s}{m} m!\right| \leq \sum_{s\in I} |B_s| m^s \leq \sum_{s\in I} \frac{m^2+1}{2} \frac{m^{2s}}{(m^2-1)^{2s}} (m^2-1) 4^s \lesssim \frac{1}{m^{\alpha m^2 - 6}}.
\]
For odd $s\not\in I$ it holds that $B_s > 0$. Therefore
\begin{align*}
G_2 &= \sum_{s\geq m} B_s \stir{s}{m} m! \\
&= \sum_{s=m}^{\alpha m^2}  B_s \stir{s}{m} m! + \sum_{\substack{s>\alpha m^2\\ s\not\in I}} B_s\stir{s}{m}m! + \sum_{s\in I} B_s\stir{s}{m}m!\\
&\gtrsim \frac{1}{m^2 e^m} + 0 - \frac{1}{m^{\alpha m^2-6}}\\
&\gtrsim \frac{1}{m^2 e^m}
\end{align*}
provided that $m$ is large enough.

Next we analyse the contribution from $r_1(k)$. Plugging \eqref{eqn:r_1(k)} into
\[
(m^2+x)\ln(m^2+x) = m^2\ln m^2 + x\ln m^2 + x + \sum_{i = 2}^\infty (-1)^i\frac{x^i}{i(i-1)m^{2(i-1)}}
\]
gives that
\[
r_1(k)\ln r_1(k) = 2m^2\ln m^2 + (\ln m^2)\sum_j \frac{(-1)^{j-1} C_{j-1} m^j k^j}{(m^2-1)^{2j-1}} + \sum A_s k^s,
\]
where
\[
A_s = \frac{(-1)^sm^s}{s(m^2-1)^{2s}}\left( (m^2-1)\binom{2s-2}{s-1} + \sum_{i=2}^s \frac{(m^2-1)^i}{(i-1)m^{2(i-1)}}\binom{2s-i-1}{s-1} \right).
\]
Therefore
\[
G_1 = (-1)^m m!(\ln m^2)\left[\sum_{j\geq m}\stir{j}{m} \frac{(-1)^jm^j}{(m^2-1)^{2j-1}} + \sum_{s\geq m} A_s\stir{s}{m}\right]
\]
whence it follows that
\[
|G_1| \leq \ln(m^2) \left[ \sum_{j=m}^\infty \frac{(2m)^{2j}}{(m^2-1)^{2j-1}} + \sum_{s\geq m} \frac{m^s m^2 2^s}{s(m^2-1)^{2s}}\cdot m^s  \right] \lesssim \frac{\ln(m^2)}{m^{2m-2}},
\]
which is negligible compared with $G_2$. We conclude that $G_1 + G_2\neq 0$.
\end{proof}

\bibliographystyle{plain}
\bibliography{literature}

\begin{thebibliography}{10}

\bibitem{Wolfram1}
Generalized hypergeometric function {3F2}: Integral representation.
\newblock
  \texttt{http://functions.wolfram.com/HypergeometricFunctions/Hypergeometric3F2/07/01/01/0001/}.
\newblock Accessed: 29 March 2015.

\bibitem{ams99}
Noga Alon, Yossi Matias, and Mario Szegedy.
\newblock The space complexity of approximating the frequency moments.
\newblock {\em J. Comput. Syst. Sci.}, 58(1):137--147, 1999.

\bibitem{a00}
Orly Alter, Patrick~O Brown, and David Botstein.
\newblock Singular value decomposition for genome-wide expression data
  processing and modeling.
\newblock {\em Proceedings of the National Academy of Sciences}, 97(18):10101,
  2000.

\bibitem{ako11}
Alexandr Andoni, Robert Krauthgamer, and Krzysztof Onak.
\newblock Streaming algorithms via precision sampling.
\newblock In {\em Proceedings of the {IEEE} 52nd {FOCS}}, pages 363--372, 2011.

\bibitem{AKR15}
Alexandr Andoni, Robert Krauthgamer, and Ilya~P. Razenshteyn.
\newblock Sketching and embedding are equivalent for norms.
\newblock In {\em Proceedings of the 47th {ACM} {STOC}}, pages 479--488, 2015.

\bibitem{AN13}
Alexandr Andoni and Huy~L. Nguyen.
\newblock Eigenvalues of a matrix in the streaming model.
\newblock In {\em Proceedings of the 24th ACM-SIAM {SODA}}, pages 1729--1737,
  2013.

\bibitem{AAR99}
George~A. Andrews, Richard Askey, and Ranjan Roy.
\newblock {\em Special Functions}.
\newblock Cambridge University Press, 1999.

\bibitem{bjk04}
Ziv Bar{-}Yossef, T.~S. Jayram, and Iordanis Kerenidis.
\newblock Exponential separation of quantum and classical one-way communication
  complexity.
\newblock In {\em Proceedings of the 36th {ACM} {STOC}}, pages 128--137, 2004.

\bibitem{bjks04}
Ziv Bar{-}Yossef, T.~S. Jayram, Ravi Kumar, and D.~Sivakumar.
\newblock An information statistics approach to data stream and communication
  complexity.
\newblock {\em J. Comput. Syst. Sci.}, 68(4):702--732, 2004.

\bibitem{bozkurt2012sum}
{\c{S}}.~Burcu Bozkurt and Durmu{\c{s}} Bozkurt.
\newblock On the sum of powers of normalized laplacian eigenvalues of graphs.
\newblock {\em MATCH Communications in Mathematical and in Computer Chemistry},
  68(3), 2012.

\bibitem{bc14}
Vladimir Braverman and Stephen~R. Chestnut.
\newblock Streaming sums in sublinear space.
\newblock arXiv:1408.5096, 2014.

\bibitem{bc15}
Vladimir Braverman and Stephen~R. Chestnut.
\newblock Universal sketches for the frequency negative moments and other
  decreasing streaming sums.
\newblock In {\em Proceedings of {APPROX/RANDOM}}, pages 591--605, 2015.

\bibitem{bcwy15}
Vladimir Braverman, Stephen~R. Chestnut, David~P. Woodruff, and Lin~F. Yang.
\newblock Streaming space complexity of nearly all functions of one variable on
  frequency vectors.
\newblock In {\em Proceedings of {PODS}}, 2016.

\bibitem{bksv14}
Vladimir Braverman, Jonathan Katzman, Charles Seidell, and Gregory Vorsanger.
\newblock An optimal algorithm for large frequency moments using
  ${O}(n^{1-2/k})$ bits.
\newblock In {\em Proceedings of {APPROX/RANDOM}}, pages 531--544, 2014.

\bibitem{bo10a}
Vladimir Braverman and Rafail Ostrovsky.
\newblock Zero-one frequency laws.
\newblock In {\em Proceedings of the 42nd {ACM} {STOC}}, pages 281--290, 2010.

\bibitem{bor15}
Vladimir Braverman, Rafail Ostrovsky, and Alan Roytman.
\newblock Zero-one laws for sliding windows and universal sketches.
\newblock In {\em Proceedings of {APPROX/RANDOM}}, pages 573--590, 2015.

\bibitem{BS15}
Marc Bury and Chris Schwiegelshohn.
\newblock Sublinear estimation of weighted matchings in dynamic data streams.
\newblock In {\em Proceedings of 23rd {ESA}}, pages 263--274, 2015.

\bibitem{cr12}
Emmanuel~J. Cand{\`e}s and Benjamin Recht.
\newblock Exact matrix completion via convex optimization.
\newblock {\em Commun. ACM}, 55(6):111--119, 2012.

\bibitem{cks03}
Amit Chakrabarti, Subhash Khot, and Xiaodong Sun.
\newblock Near-optimal lower bounds on the multi-party communication complexity
  of set disjointness.
\newblock In {\em Proceedings of the 18th {IEEE} {CCC}}, pages 107--117, 2003.

\bibitem{ccf04}
Moses Charikar, Kevin~C. Chen, and Martin Farach{-}Colton.
\newblock Finding frequent items in data streams.
\newblock {\em Theor. Comput. Sci.}, 312(1):3--15, 2004.

\bibitem{cw09}
Kenneth~L. Clarkson and David~P. Woodruff.
\newblock Numerical linear algebra in the streaming model.
\newblock In {\em Proceedings of the 41st {ACM} {STOC}}, pages 205--214, 2009.

\bibitem{cm05}
Graham Cormode and S.~Muthukrishnan.
\newblock An improved data stream summary: the count-min sketch and its
  applications.
\newblock {\em J. Algorithms}, 55(1):58--75, 2005.

\bibitem{dtv11}
Amit Deshpande, Madhur Tulsiani, and Nisheeth~K. Vishnoi.
\newblock Algorithms and hardness for subspace approximation.
\newblock In {\em Proceedings of SODA}, pages 482--496, 2011.

\bibitem{d14}
Xuan~Vinh Doan and Stephen Vavasis.
\newblock Finding the largest low-rank clusters with {K}y {F}an $2$-$k$-norm
  and $\ell_1$-norm.
\newblock arXiv:1403.5901, 2014.

\bibitem{KK02}
K.~Driver and K.~Jordan.
\newblock Zeroes of the hypergeometric polynomials ${F}(-n,b;c;z)$.
\newblock {\em Proceedings of Algorithms for Approximations}, IV:436--445,
  2002.

\bibitem{g12}
Sumit Ganguly.
\newblock A lower bound for estimating high moments of a data stream.
\newblock arXiv:1201.0253, 2012.

\bibitem{g15}
Sumit Ganguly.
\newblock Taylor polynomial estimator for estimating frequency moments.
\newblock In {\em Proceedings of 42th {ICALP}}, pages 542--553, 2015.

\bibitem{gc07}
Sumit Ganguly and Graham Cormode.
\newblock On estimating frequency moments of data streams.
\newblock In {\em Proceedings of {RANDOM/APPROX}}, pages 479--493, 2007.

\bibitem{GKKRW07}
Dmitry Gavinsky, Julia Kempe, Iordanis Kerenidis, Ran Raz, and Ronald de~Wolf.
\newblock Exponential separations for one-way quantum communication complexity,
  with applications to cryptography.
\newblock In {\em Proceedings of the 39th {ACM} {STOC}}, pages 516--525, 2007.

\bibitem{GD14}
Matan Gavish and David~L. Donoho.
\newblock Optimal shrinkage of singular values.
\newblock Technical Report 2014-08, Department of Statistics, Stanford
  University, May 2014.

\bibitem{g09}
Andre Gronemeier.
\newblock Asymptotically optimal lower bounds on the {NIH}-multi-party
  information complexity of the {AND}-function and disjointness.
\newblock In {\em Proceedings of the 26th {STACS}}, pages 505--516, 2009.

\bibitem{hlm10}
Moritz Hardt, Katrina Ligett, and Frank McSherry.
\newblock A simple and practical algorithm for differentially private data
  release.
\newblock In {\em Advances in Neural Information Processing Systems 25}, pages
  2348--2356. 2012.

\bibitem{i06}
Piotr Indyk.
\newblock Stable distributions, pseudorandom generators, embeddings, and data
  stream computation.
\newblock {\em J. {ACM}}, 53(3):307--323, 2006.

\bibitem{iw03}
Piotr Indyk and David~P. Woodruff.
\newblock Tight lower bounds for the distinct elements problem.
\newblock In {\em Proceedings of the 44th {IEEE} {FOCS}}, pages 283--288, 2003.

\bibitem{iw05}
Piotr Indyk and David~P. Woodruff.
\newblock Optimal approximations of the frequency moments of data streams.
\newblock In {\em Proceedings of the 37th {ACM} {STOC}}, pages 202--208, 2005.

\bibitem{j09}
T.~S. Jayram.
\newblock Hellinger strikes back: {A} note on the multi-party information
  complexity of {AND}.
\newblock In {\em Proceedings of {RANDOM}/{APPROX}}, pages 562--573, 2009.

\bibitem{JST11}
Hossein Jowhari, Mert Saglam, and G{\'{a}}bor Tardos.
\newblock Tight bounds for ${L}_p$ samplers, finding duplicates in streams, and
  related problems.
\newblock In {\em Proceedings of the 30th {ACM} {SIGMOD} {PODS}}, pages 49--58,
  2011.

\bibitem{ks92}
Bala Kalyanasundaram and Georg Schnitger.
\newblock The probabilistic communication complexity of set intersection.
\newblock {\em {SIAM} J. Discrete Math.}, 5(4):545--557, 1992.

\bibitem{knpw11}
Daniel~M. Kane, Jelani Nelson, Ely Porat, and David~P. Woodruff.
\newblock Fast moment estimation in data streams in optimal space.
\newblock In {\em STOC}, pages 745--754, 2011.

\bibitem{knw10}
Daniel~M. Kane, Jelani Nelson, and David~P. Woodruff.
\newblock On the exact space complexity of sketching and streaming small norms.
\newblock In {\em Proceedings of the 21st {ACM-SIAM} {SODA}}, pages 1161--1178,
  2010.

\bibitem{kv16}
Weihao Kong and Gregory Valiant.
\newblock Spectrum estimation from samples.
\newblock arXiv:1602.00061, 2016.

\bibitem{ks03}
Robert Krauthgamer and Ori Sasson.
\newblock Property testing of data dimensionality.
\newblock In {\em Proceedings of the Fourteenth {ACM-SIAM} {SODA}}, pages
  18--27, 2003.

\bibitem{lm12}
Chao Li and Gerome Miklau.
\newblock Optimal error of query sets under the differentially-private matrix
  mechanism.
\newblock In {\em Proceedings of {ICDT}}, pages 272--283, 2013.

\bibitem{LNW14}
Yi~Li, Huy~L. Nguyen, and David~P. Woodruff.
\newblock On sketching matrix norms and the top singular vector.
\newblock In {\em Proceedings of the 25th {ACM-SIAM} {SODA}}, pages 1562--1581,
  2014.

\bibitem{lnw14b}
Yi~Li, Huy~L. Nguyen, and David~P. Woodruff.
\newblock Turnstile streaming algorithms might as well be linear sketches.
\newblock In {\em Proceedings of the 46th {ACM} {STOC}}, pages 174--183, 2014.

\bibitem{lww14}
Yi~Li, Zhengyu Wang, and David~P. Woodruff.
\newblock Improved testing of low rank matrices.
\newblock In {\em Proceedings of the 20th {ACM} {SIGKDD}}, pages 691--700,
  2014.

\bibitem{lw13}
Yi~Li and David~P. Woodruff.
\newblock A tight lower bound for high frequency moment estimation with small
  error.
\newblock In {\em Proceedings of {RANDOM}/{APPROX}}, pages 623--638, 2013.

\bibitem{MW10}
Morteza Monemizadeh and David~P. Woodruff.
\newblock 1-pass relative-error ${L}_p$-sampling with applications.
\newblock In {\em SODA}, pages 1143--1160, 2010.

\bibitem{nw10}
Jelani Nelson and David~P. Woodruff.
\newblock Fast {M}anhattan sketches in data streams.
\newblock In {\em PODS}, pages 99--110, 2010.

\bibitem{r92}
Alexander~A. Razborov.
\newblock On the distributional complexity of disjointness.
\newblock {\em Theor. Comput. Sci.}, 106(2):385--390, 1992.

\bibitem{R15}
Ilya Razenshteyn.
\newblock Personal communication, 2015.

\bibitem{SF95}
Robert Sedgewick and Phillipe Flajolet.
\newblock {\em An Introduction to the Analysis of Algorithms}.
\newblock Addison-Wesley, 1996.

\bibitem{s09}
Dragan Stevanovic, Aleksandar Ilic, Cristina Onisor, and Mircea~V Diudea.
\newblock {LEL}--a newly designed molecular descriptor.
\newblock {\em Acta Chim. Slov}, 56:410--417, 2009.

\bibitem{vy11}
Elad Verbin and Wei Yu.
\newblock The streaming complexity of cycle counting, sorting by reversals, and
  other problems.
\newblock In {\em Proceedings of the 22nd {ACM-SIAM} {SODA}}, pages 11--25,
  2011.

\bibitem{w04}
David~P. Woodruff.
\newblock Optimal space lower bounds for all frequency moments.
\newblock In {\em SODA}, pages 167--175, 2004.

\bibitem{w14}
David~P. Woodruff.
\newblock Sketching as a tool for numerical linear algebra.
\newblock {\em Foundations and Trends in Theoretical Computer Science},
  10(1-2):1--157, 2014.

\bibitem{wz12}
David~P. Woodruff and Qin Zhang.
\newblock Tight bounds for distributed functional monitoring.
\newblock In {\em Proceedings of the 44th {ACM} {STOC}}, pages 941--960, 2012.

\bibitem{Xia}
Dong Xia.
\newblock Optimal schatten-$q$ and {K}y-{F}an-$k$ norm rate of low rank matrix
  estimation, 2014.

\bibitem{Zhang95}
Zhengyou Zhang.
\newblock Parameter estimation techniques: a tutorial with application to conic
  fitting.
\newblock {\em Image and Vision Computing}, 15(1):59--76, 1997.

\bibitem{z08}
Bo~Zhou.
\newblock On sum of powers of the {L}aplacian eigenvalues of graphs.
\newblock {\em Linear Algebra and its Applications}, 429(8--9):2239--2246,
  2008.

\end{thebibliography}
\end{document}